\documentclass[12pt]{article}
\usepackage{amsmath,amssymb,amscd,amsthm,color,graphics,verbatim}
\usepackage[margin=0.95in]{geometry}

\usepackage[enableskew]{youngtab}
\def\qed{\hfill $\vrule height 2.5mm  width 2.5mm depth 0mm $}

\newtheorem{theorem}{Theorem}
\newtheorem{proposition}[theorem]{Proposition}
\newtheorem{lemma}[theorem]{Lemma}

\theoremstyle{definition}

\newtheorem{definition}[theorem]{Definition}

\newtheorem{example}[theorem]{Example}

\begin{document}
$\,$\vspace{5mm}

\begin{center}
{\sf\huge Rigged Configurations Approach For The}
\vspace{3mm}\\
{\sf\huge Spin-1/2 Isotropic Heisenberg Model}
\vspace{20mm}\\
{\textsf{\huge  Reiho Sakamoto}}
\vspace{2mm}\\
{\textsf {Department of Physics,}}
\vspace{-1mm}\\
{\textsf {Tokyo University of Science, Kagurazaka,}}
\vspace{-1mm}\\
{\textsf {Shinjuku, Tokyo, 162-8601, Japan}}
\vspace{-1mm}\\
{\textsf {reiho@rs.tus.ac.jp}}
\vspace{-1mm}\\
{\textsf {reihosan@08.alumni.u-tokyo.ac.jp}}
\vspace{50mm}
\end{center}

\begin{abstract}
\noindent
We continue the rigged configurations analysis of the solutions to the Bethe ansatz equations
for the spin-1/2 isotropic Heisenberg model.
We analyze the non self-conjugate strings of Deguchi--Giri and
the counter example for the string hypothesis discovered by Essler--Korepin--Schoutens.
In both cases clear discrete structures appear.
\end{abstract}

\pagebreak

\section{Introduction}
Bethe's famous solution \cite{Bethe} to the spin-1/2 isotropic Heisenberg model under
the periodic boundary condition is one of the prototypical examples of the quantum integrable models.
In this problem, our first goal is to find all eigenvectors and eigenvalues of the Hamiltonian
so that we can describe the complete set of solutions for the corresponding Schr\"{o}dinger equation.
Basically Bethe's method is comprised of two steps.
First we assume a specific form of the eigenstates which depend on parameters $\lambda_j$.
We call them the Bethe vectors.
Then we determine the parameters $\lambda_j$ by solving the so-called
Bethe ansatz equations (see Eq.(\ref{eq:Bethe_ansatz}) for the precise form).
Later Bethe's method was reformulated in more algebraic form by Faddeev's school
(see \cite{Faddeev,KorepinBook}).

Despite of a huge amount of works that have appeared over several decades,
there are still remaining unclear aspects about the Bethe ansatz theory.
When we try to construct the eigenvectors according to the Bethe ansatz,
we first have to solve the Bethe ansatz equations.
Here we usually assume that it is enough to consider pairwise distinct solutions to the equations.
Then one immediately notices that there are too many solutions to the Bethe ansatz equations
and only some of them correspond to non-zero Bethe vectors
(see, for example, \cite{HNS1} and \cite[Section 4.2]{KiSa14}).
We call the solutions corresponding to non-zero Bethe vectors as regular solutions
and the remaining solutions as singular solutions.

Unfortunately the set of regular solutions is too small to construct all the eigenvectors.
Therefore it is natural to find a method to construct the eigenvectors
from a subset of the singular solutions.
A particularly convenient method is to introduce certain regularization
(see Eq.(\ref{eq:regularization_of_Nepomechie})) to
the solutions to the Bethe ansatz equations \cite{AV,Beisert}.
Recently, Nepomechie--Wang \cite{NepomechieWang2013}
worked out the procedure in general and proposed a criterion
to select the subset of the singular solutions which generate the remaining eigenvectors
after the regularization (see also \cite{NepomechieWang2014} for an alternative derivation).
This proposal is checked by an extensive numerical computation \cite{HNS1}.
Also the corresponding eigenvalues are obtained in \cite{KirillovSakamoto2014b}.
See Section \ref{sec:singular_solution} and Section \ref{sec:appendix}
for a somewhat extensive review of these developments.
Therefore we now have a conjectural characterization of the solutions to the Bethe ansatz equations
which shall generate all the eigenvectors.

Despite of these developments, not much is known about the structure of
the solutions to the Bethe ansatz equations.
Let us call the union of the regular solutions and Nepomechie--Wang's
physical singular solutions as physical solutions to the Bethe ansatz equations.
A popular method for the analysis of the physical solutions is
the so-called string hypothesis \cite{Bethe,Takahashi}.
In this hypothesis, we assume that the roots of physical solutions appear in the form of strings
\begin{align}
\label{eq:string}
a+bi,\,a+(b-1)i,\,a+(b-2)i,\,\ldots,\,a-bi,\qquad
(a\in\mathbb{R},\,b\in\mathbb{Z}_{\geq 0}/2).
\end{align}
Based on this assumption, we derive (half) integers called the Bethe--Takahashi
quantum numbers (see \cite[Eq.(2.11)]{Takahashi})
which label the strings in a solution (see, for example, \cite{HC,SD} for detailed explanations).
However, Essler--Korepin--Schoutens \cite{EKS:1992} discovered a counter example to this scenario.
Namely, for the 2-down spin sector of length more than 21.86 chains, they discovered a pair of real roots
which have the same Bethe--Takahashi quantum number, thus we cannot apply
the Bethe--Takahashi quantum number for classification of them.
Therefore we think it is worthwhile to consider a possible method to provide
well-defined quantum numbers to classify strings in the solutions to the Bethe ansatz equations.

In this paper, we elaborate the rigged configurations approach for this problem \cite{KiSa14}
by analyzing much more complicated solutions.
The rigged configurations approach is based on the following observation.
We introduce an order within the set of physical solutions of the same strings type
based on the real parts of the longest strings.
Then, from the relative positions of the strings, we construct a bijection
with combinatorial objects called the rigged configurations.
We apply this method for two classes of solutions.
One is the non self-conjugate strings due to Deguchi--Giri \cite{DG}
(see Section \ref{sec:deguchi_giri}).
The non self-conjugate string is a set of roots which is not a single string,
nevertheless we cannot split it into smaller strings,
see Eq.(\ref{eq:3+1_non_self-conjugate-strings}) for example.
Based on the rigged configurations approach we provide a clear demonstration
that the non self-conjugate strings are obtained by a fusion of smaller strings
even if they look irregular.

Another solution which we analyze in the present paper is the counter example
to the string hypothesis at 2-down spin sector of length 25 chains discovered by
Essler--Korepin--Schoutens \cite{EKS:1992} (see Section \ref{sec:korepin}).
We analyze 255 real solutions to the corresponding Bethe ansatz equations
by the rigged configurations approach.
Then we naturally find two exceptional real solutions which violate the rigged configurations structure.
Next we make correspondence between the set of a combination of two exceptional real solutions
and 20 physical complex solutions and the set of the rigged configurations.
All physical solutions in this case naturally fit into the rigged configurations scheme.
As the result we can determine the quantum numbers (we call riggings)
of all physical solutions in the present case.
Our result agrees with the interpretation by \cite{EKS:1992}
except for the exceptional real solutions.

In both cases, we observe clear discrete structures within the solutions to the Bethe ansatz equations.
Underlying our philosophy is that the Bethe ansatz method for the spin-1/2
isotropic Heisenberg model is a well-defined mathematical procedure.

This paper is organized as follows.
In Section \ref{sec:bethe_ansatz}, we briefly review standard results from
the algebraic Bethe ansatz analysis of the spin-1/2 isotropic Heisenberg model.
In Section \ref{sec:singular_solution}, we review recent developments about
the theory of the singular solutions.
In Section \ref{sec:rigged_configurations} we review necessary facts about the rigged configurations.
In Section \ref{sec:deguchi_giri} we analyze the non self-conjugate strings.
In Section \ref{sec:korepin} we analyze the counter example to the string hypothesis.
In Section \ref{sec:appendix} we provide complete proofs for the assertions in
Section \ref{sec:singular_solution} according to the line described
in \cite{KirillovSakamoto2014b}.
In Section \ref{sec:update} we provide an update for our previous paper \cite{KiSa14}.

All numerical results in this work are due to Mathematica.

\section{Algebraic Bethe ansatz}
\label{sec:bethe_ansatz}
In this section, we briefly review standard facts from the algebraic Bethe ansatz analysis for
the spin-1/2 isotropic Heisenberg model with the periodic boundary condition
(see \cite{Faddeev,KorepinBook} for the details).
In the length $N$ chain case, the Hamiltonian is
\begin{align}
\mathcal{H}_N &=  \frac{J}{4} \sum_{k=1}^{N}( \sigma_{k}^{1} \sigma_{k+1}^{1}
+ \sigma_{k}^{2} \sigma_{k+1}^{2}+ \sigma_{k}^{3} \sigma_{k+1}^{3} -{\mathbb {I}}_N),\qquad
\sigma_{N+1}^{a}=\sigma_{1}^{a}
\end{align}
where $\sigma^a$ ($a=1,2,3$) are the Pauli matrices and
$\mathbb{I}_N=I^{\otimes N}$ with the $2\times 2$ identity matrix $I$.
The space of states of our systems is
$\mathfrak{H}_N = \bigotimes_{j=1}^{N} V_j$ $(V_j \simeq {\mathbb {C}}^2)$.
Then $\sigma^a_k$ acts as $\sigma^a$ on the $k$-th site $V_k$
and as the identity operator on the rest of sites.
Below we aim to construct the solutions to the Schr\"{o}dinger equation
$\mathcal{H}_N\Psi_\lambda=\mathcal{E}_\lambda\Psi_\lambda$
where $\mathcal{E}_\lambda$ is the energy eigenvalue.

A basic apparatus of the algebraic Bethe ansatz is the transfer matrices.
For this purpose, we introduce the so-called Lax operators
\begin{align}
\label{eq:def_L_of_Bethe}
L_{k}(\lambda)=
\left(\!
\begin{array}{cc}
\lambda \mathbb{I}_N+\frac{i}{2}\sigma^3_k & \frac{i}{2}\sigma^-_k\\
\frac{i}{2}\sigma^+_k & \lambda \mathbb{I}_N-\frac{i}{2}\sigma^3_k
\end{array}
\!\right),\qquad
(k=1,2,\ldots,N,\,\sigma^\pm_k=\sigma^1_k\pm i\sigma^2_k)
\end{align}
which act on the space $\mathbb{C}^2_0\otimes\mathfrak{H}_N$.
Here the first $\mathbb{C}^2_0$ is the auxiliary space $\mathbb{C}^2$ which provides
the above $2\times 2$ matrix presentation of $L_{k}(\lambda)$.
The matrix entries of $L_{k}(\lambda)$ act on $\mathfrak{H}_N$.
Then the transfer matrix is
\begin{align}
\label{def:transfer_matrix}
T_N(\lambda)=L_N(\lambda)L_{N-1}(\lambda)\cdots L_1(\lambda).
\end{align}
From the $2\times 2$ matrix presentation of $T_N(\lambda)$
\begin{align}
T_N(\lambda)=
\left(
\begin{array}{cc}
A_N(\lambda) & B_N(\lambda)\\
C_N(\lambda) & D_N(\lambda)
\end{array}
\right),
\end{align}
we define the operators $A_N(\lambda)$, $B_N(\lambda)$, $C_N(\lambda)$ and $D_N(\lambda)$
which act on the space of states $\mathfrak{H}_N$.
Let $\tau_N(\lambda)=\operatorname{tr}|_{\mathbb{C}^2_0}T_N(\lambda)
=A_N(\lambda)+D_N(\lambda)$.
Then the fundamental formula in the algebraic Bethe ansatz is
\begin{align}
\label{eq:tau_hamiltonian}
\mathcal{H}_N=
\frac{iJ}{2}\frac{d}{d\lambda}\log \tau_N(\lambda)\Bigr|_{\lambda=\frac{i}{2}}
-\frac{NJ}{2}\mathbb{I}_N.
\end{align}
Therefore, in order to obtain the eigenvectors and the eigenvalues of the Hamiltonian $\mathcal{H}_N$,
it is enough to diagonalize $\tau_N(\lambda)$.
Let us start from the spacial vector
$|0\rangle_N=\left(\!
\begin{array}{c}
1\\0
\end{array}\!
\right)^{\otimes N}\in\mathfrak{H}_N$.
Note that the vector $|0\rangle_N$ is an eigenvector of $\tau_N(\lambda)$.
On this vector, we apply $B_N(\lambda)$ as the creation operator.
We call the following vectors the Bethe vectors.
\begin{align}
\Psi_N(\lambda_1,\cdots,\lambda_\ell)
=B_N(\lambda_1)\cdots B_N(\lambda_\ell)|0\rangle_N.
\end{align}
Since the operators $B_N$ are commutative
(see (\ref{eq:Bethe_relations1})), we see that the order of the parameters
$\lambda_1,\cdots,\lambda_\ell$ is not important.
It is known that $\ell$ is the number of down spins in the state $\Psi_N$.
Then the main construction is as follows.

\begin{theorem}
The Bethe vector $\Psi_N(\lambda_1,\cdots,\lambda_\ell)$
is an eigenvector of the transfer matrix $\tau_N(\lambda)$
if and only if the numbers $\lambda_1,\cdots,\lambda_\ell$
satisfy the following system of algebraic equations
\begin{align}
\label{eq:Bethe_ansatz}
\left(
\frac{\lambda_k+\frac{i}{2}}{\lambda_k-\frac{i}{2}}
\right)^N
=\prod_{j=1 \atop j\neq k}^\ell
\frac{\lambda_k-\lambda_j+i}{\lambda_k-\lambda_j-i},
\qquad
(k=1,\cdots,\ell).
\end{align}
These equations are called the Bethe ansatz equations.
\qed
\end{theorem}

By a direct computation, we see that non-zero Bethe vector
$\Psi_N(\lambda_1,\cdots,\lambda_\ell)$ is the eigenvector of the Hamiltonian $\mathcal{H}_N$
with the energy eigenvalue $\mathcal{E}=-\frac{J}{2}\sum_{j=1}^\ell
\frac{1}{\lambda_j^2+\frac{1}{4}}$.
If a pairwise distinct solution $(\lambda_1,\cdots,\lambda_\ell)$ to the Bethe ansatz equations
corresponds to a non-zero Bethe vector, we call it a {\bf regular solution}.

Note that we can introduce a natural action of the Lie algebra $\mathfrak{sl}_2$
on the space $\mathfrak{H}_N$.
Then it is known that the Bethe vectors thus obtained correspond to the highest weight vectors.
The other eigenvectors of the Hamiltonian are obtained by the lowering operators of $\mathfrak{sl}_2$.
Nevertheless, it is known for a long time that the regular solutions to the Bethe ansatz equations
do not cover all the vectors of $\mathfrak{H}_N$.
This will be the subject of the next section.

\section{Singular solutions to the Bethe ansatz equations}
\label{sec:singular_solution}
As we see in the previous section, if $\{\lambda_1,\ldots,\lambda_\ell\}$ is a solution to
the Bethe ansatz equations, we have
\begin{align}
\mathcal{H}_N\Psi_N(\lambda_1,\ldots,\lambda_\ell)=
\mathcal{E}_{\lambda_1,\ldots,\lambda_\ell}
\Psi_N(\lambda_1,\ldots,\lambda_\ell),\qquad
\mathcal{E}_{\lambda_1,\ldots,\lambda_\ell}:=-\frac{J}{2}\sum_{j=1}^\ell\frac{1}{\lambda_j^2+\frac{1}{4}}.
\end{align}
Therefore the following type of the solutions correspond to divergent energy eigenvalues.
\begin{definition}
The pairwise distinct solutions to the Bethe ansatz equations of the form
\begin{align}
\label{eq:general_singular_solution}
\left\{
\frac{i}{2},-\frac{i}{2},\lambda_3,\ldots,\lambda_\ell
\right\}.
\end{align}
are called {\bf singular}.
\qed
\end{definition}
The Bethe vectors corresponding to singular solutions are the null vector.
However, in the following arguments based on the rigged configurations,
it is necessarily to analyze all physically relevant solutions to the Bethe ansatz equations.
Therefore we need to introduce a certain regularization scheme corresponding to the singular solutions.
We remark that this problem was already noticed in Bethe's original paper \cite{Bethe}.

Following Nepomechie--Wang \cite{NepomechieWang2013},
let us consider the following regularization
\begin{align}
\label{eq:regularization_of_Nepomechie}
\lambda_1=\frac{i}{2}+\epsilon+c\,\epsilon^N,\qquad
\lambda_2=-\frac{i}{2}+\epsilon.
\end{align}
As we shall see in the sequel, this regularization leads to finite Bethe vectors.
Moreover, Kirillov--Sakamoto \cite{KirillovSakamoto2014b} derived
the corresponding energy eigenvalues under the regularization.
Note that this type of the regularization was also considered by Avdeev--Vladimirov \cite{AV} in 1987
and Beisert et. al. \cite[equation (3.4)]{Beisert} in 2003.
In the case of $\ell=2$, the resulting eigenvector coincides with
the explicit expression of \cite[equation (26)]{EKS:1992}.
See \cite{GD} where the authors derived several explicit expressions for the singular eigenvectors,
including $\ell=2$ case, under the same regularization considered here
based on the result of \cite{Deguchi2001}.

In general, our regularization scheme proceeds as follows.
We start from the general singular solution (\ref{eq:general_singular_solution}).
If a singular solution satisfies the following criterion due to Nepomechie--Wang
\begin{align}
\label{eq:NepomechieWangCriterion}
\left(
-\prod_{j=3}^\ell
\frac{\lambda_j+\frac{i}{2}}{\lambda_j-\frac{i}{2}}
\right)^N=1,
\end{align}
we call {\bf physical singular solution}.
We set the number $c$ of (\ref{eq:regularization_of_Nepomechie}) as follows
\begin{align}
c=2i^{N+1}\prod_{j=3}^\ell
\frac{\lambda_j+\frac{3i}{2}}{\lambda_j-\frac{i}{2}}.
\end{align}
According to \cite{NepomechieWang2013},
\begin{align}
\lim_{\epsilon\rightarrow 0}\frac{1}{\epsilon^N}
B_N\!\left(\frac{i}{2}+\epsilon+c\,\epsilon^N\right)B_N\!\left(-\frac{i}{2}+\epsilon\right)
B_N(\lambda_3)\cdots B_N(\lambda_\ell)
|0\rangle_N
\end{align}
gives a well-defined eigenvector of the Hamiltonian with the eigenvalue \cite{KirillovSakamoto2014b}
\begin{align}
\mathcal{E}_{\frac{i}{2},-\frac{i}{2},\lambda_3,\ldots,\lambda_\ell}
=-J-\frac{J}{2}\sum_{j=3}^\ell\frac{1}{\lambda_j^2+\frac{1}{4}}.
\end{align}

Then the main conjecture of \cite{NepomechieWang2013} states that the Bethe vectors
corresponding to regular solutions and physical singular solutions provide
the complete set of the highest weight vectors of the model.
Hao--Nepomechie--Sommese \cite{HNS1} computed all pairwise distinct solutions
to the Bethe ansatz equations up to $N=14$ and confirmed this conjecture.
Therefore we believe this conjecture is very likely to be true.

\begin{example}
When $N=4$, we have the singular solution $\{i/2,-i/2\}$.
By an explicit computation, we obtain
\begin{align}
\label{eq:NepomechieWang_N=4}
\lim_{\epsilon\rightarrow 0}\frac{1}{\epsilon^4}
B_4\!\left(\frac{i}{2}+\epsilon+c\,\epsilon^4\right)B_4\!\left(-\frac{i}{2}+\epsilon\right)|0\rangle_4
=(0,0,0,2,0,0,-2,0,0,ic,0,0,2,0,0,0)^t.
\end{align}
If we put $c=2i$, we obtain the correct eigenvector with the energy eigenvalue $-J$.
\qed
\end{example}
Note that the above example forms the initial step of the inductive proof for the assertions
in the present section which is given in the appendix at the end of the present paper.

\section{Rigged configurations}
\label{sec:rigged_configurations}
Rigged configurations are combinatorial objects introduced by \cite{KKR,KR}.
In \cite{KiSa14}, it was pointed out that the set of regular and physical singular solutions
to the Bethe ansatz equations have natural correspondence with the rigged configurations.
In the next sections we aim to elaborate this claim by analyzing more complicated solutions.
For this purpose we review basic definitions about the rigged configurations following \cite{KiSa14}.

The rigged configurations are comprised of set of data $\mu$ and $(\nu,J)$.
$\mu$ is a sequence of positive integers which specifies the shape of the space of states.
If we consider the spin $s$-Heisenberg model, the space of states is $(\mathbb{C}^{2s+1})^{\otimes N}$.
Then we set $\mu=(\overbrace{2s,\ldots,2s}^N)$.
As for $(\nu,J)$, $\nu$ is a partition (Young diagram) and $J$ is a sequence of integers associated with
each row of the diagram $\nu$ which we impose certain defining conditions.
We regard $(\nu,J)$ as a multiset $(\nu,J)=\{(\nu_1,J_1),\ldots,(\nu_g,J_g)\}$
where $g$ is the length of the partition $\nu$.
In the definition of the rigged configurations we do not make distinction about the order within
the multiset $(\nu,J)$.
For $k\in\mathbb{Z}_{\geq 0}$, define the {\bf vacancy numbers} $P_k(\nu)$ by
\begin{align}
P_k(\nu)=\sum_{j=1}^N\min(k,\mu_j)-2\sum_{j=1}^g\min(k,\nu_j).
\end{align}
Note that although we do not explicitly write $\mu$ in $P_k(\nu)$,
the vacancy numbers do depend on $\mu$.
The condition that the data $\mu$ and $(\nu,J)$ to be a rigged configuration is as follows:
\begin{align}
0\leq J_i\leq P_{\nu_i}(\nu),\qquad (0\leq i\leq g).
\end{align}
This condition implicitly requires $0\leq P_{\nu_i}(\nu)$ as the condition on $\nu$,
in which case the partition $\nu$ is called admissible.
We call $\nu$ the configuration and $J_i$ $(1\leq i\leq g)$ the rigging.

In the spin-1/2 case, the situation becomes simple.
In this case, we have $\mu=(1^N)$, thus the vacancy number is
\begin{align}
P_k(\nu)=N-2\sum_{j=1}^g\min(k,\nu_j).
\end{align}
Here $\sum_{j=1}^g\min(k,\nu_j)$ is the number of the boxes within the first $k$ columns of
the Young diagram $\nu$.
Then the partition $\nu$ is admissible when the number of the boxes $|\nu|$ satisfies $|\nu|\leq N/2$.
In the application to the Bethe ansatz, $|\nu|$ is the number of down spins $\ell$.
For example, when $N=14$, $\ell=7$,
$(\nu,J)=\{(3,0),(2,1),(1,4),(1,2)\}$ is a rigged configuration.
We use the following diagram to represent this rigged configuration:
\begin{center}
\unitlength 12pt
\begin{picture}(5,4)
\put(0,3.1){0}
\put(0,2.1){2}
\put(0,1.1){6}
\put(0,0.1){6}
\put(0.8,0){\Yboxdim12pt\yng(3,2,1,1)}
\put(4.1,3.1){0}
\put(3.1,2.1){1}
\put(2.1,1.1){4}
\put(2.1,0.1){2}
\end{picture}
\end{center}
Here we depict $\nu$ as the Young diagram and we put the riggings (resp. vacancy numbers)
on the right (resp. left) of the corresponding rows $\nu_i$.

For higher spin cases the situation becomes more subtle.
For spin-1 case, we have $\mu=(2^N)$.
Let $N=2$. Then $\nu=(1)$ and $(2)$ are admissible, though $\nu=(1,1)$ is not admissible
since we have $P_1((1,1))=\{\min(1,1)+\min(1,1)\}-2\{\min(1,1)+\min(1,1)\}=-2$.

Basic idea of the correspondence between the solutions to the Bethe ansatz equations
and the rigged configurations is as follows.
In many cases, roots of the solutions to the Bethe ansatz equations
take the following forms:
\[
a+bi,\,a+(b-1)i,\,a+(b-2)i,\,\ldots,\,a-bi,\qquad
(a\in\mathbb{R},\,b\in\mathbb{Z}_{\geq 0}/2).
\]
We call this as $(2b+1)$-string.
Then we correspond a length $(2b+1)$ row of $\nu$.
The riggings specify the relative positions of the $(2b+1)$-strings.
Below we explain how to realize this correspondence based on concrete examples.

\section{Non self-conjugate strings}
\label{sec:deguchi_giri}
Deguchi and Giri \cite{DG} reported a new type of solutions which they call non self-conjugate strings.
This class of solutions first appears when $N=12$ and its typical structure looks as follows:
\begin{align}\label{eq:3+1_non_self-conjugate-strings}
\{a+i,a+\delta_1+i\delta_2,a+\delta_1-i\delta_2,a-i\}
\end{align}
where $a\in\mathbb{R}$ and the two real numbers $\delta_1$ and $\delta_2$
are supposed to be small.
A complicated feature of this type of solutions is that it is difficult to split the string of
solutions into distinct strings of solutions.
In this section, we use the rigged configurations to analyze the
non self-conjugate strings.\footnote{For the numerical computations of the case $N=12$ and $\ell=6$,
we greatly benefited from the supplementary tables of \cite{HNS1}.}

Following \cite{KiSa14} we construct a bijection between the rigged configurations
and the solutions to the Bethe ansatz equations for $N=12$ and $\ell=6$ case.
Here we concentrate on the physical solutions which contain three strings of lengths 3, 2 and 1.
This case corresponds to the rigged configurations of the following type:
\begin{center}
\unitlength 12pt
\begin{picture}(3,3)
\put(0,2.1){0}
\put(0,1.1){2}
\put(0,0.1){6}
\put(0.8,0){\Yboxdim12pt\yng(3,2,1)}
\put(4,2.1){0}
\put(3,1.1){$r_1$}
\put(2,0.1){$r_2$}
\end{picture}
\end{center}
For the sake of the brevity we call the above rigged configurations by $(r_1,r_2)$.
Since $0\leq r_1\leq 2$ and $0\leq r_2\leq 6$ we have 21 such rigged configurations.
Now we introduce the order for 21 solutions with length 3, 2 and 1 strings
according to the real parts of the 3-strings.
Then we have a clear classification as follows (asterisk for the solutions number
means non self-conjugate strings).
We remark that the solution $\#11^*$ is a physical singular solution.
In the following diagrams, the solutions are depicted on the complex plane
and the spacing between the dotted lines is $0.5$.

\paragraph{Group 1.}
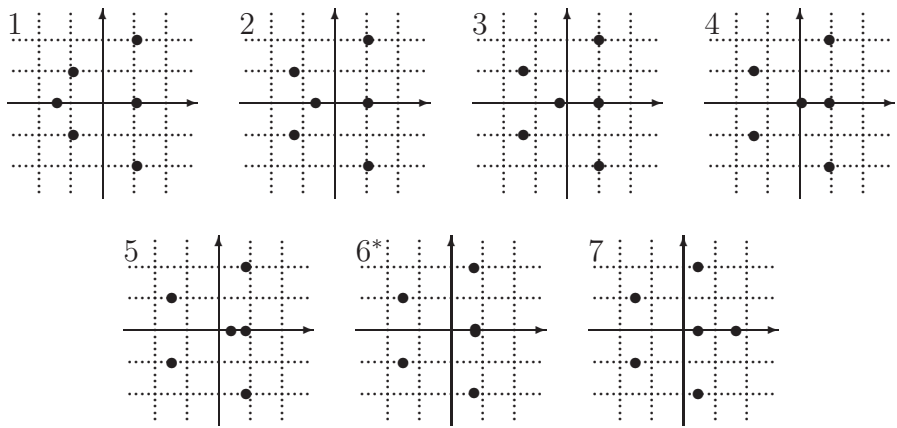
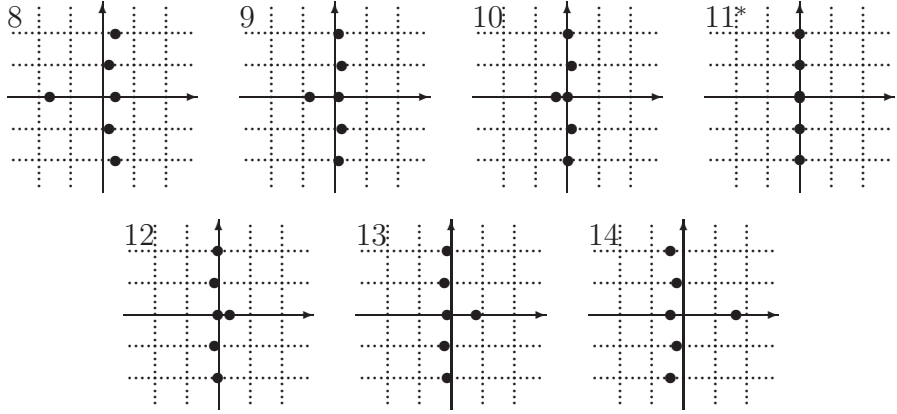
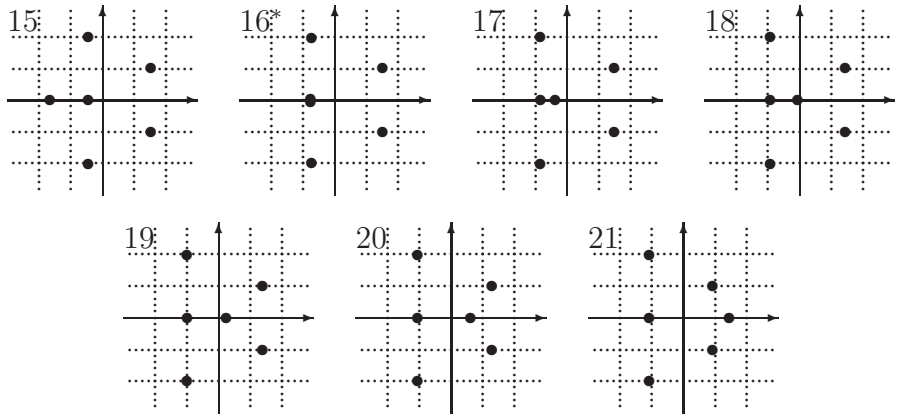
\begin{center}
\unitlength 12pt
\begin{picture}(7,6)
\put(0.0,5.2){$1$}
\put(4.084,3.){\circle*{0.3}}
\put(4.088,4.992){\circle*{0.3}}
\put(2.084,4.){\circle*{0.3}}
\put(1.57,3.){\circle*{0.3}}
\put(2.084,2.){\circle*{0.3}}
\put(4.088,1.008){\circle*{0.3}}
\put(0,3){\vector(1,0){6}}
\put(3,0){\vector(0,1){6}}
\multiput(1,0.21)(0,0.2){29}{\circle*{0.07}}
\multiput(2,0.21)(0,0.2){29}{\circle*{0.07}}
\multiput(4,0.21)(0,0.2){29}{\circle*{0.07}}
\multiput(5,0.21)(0,0.2){29}{\circle*{0.07}}
\multiput(0.2,1.01)(0.2,0){29}{\circle*{0.07}}
\multiput(0.2,2.01)(0.2,0){29}{\circle*{0.07}}
\multiput(0.2,4.01)(0.2,0){29}{\circle*{0.07}}
\multiput(0.2,5.01)(0.2,0){29}{\circle*{0.07}}
\end{picture}
\begin{picture}(7,6)
\put(0.0,5.2){$2$}
\put(4.054,3.){\circle*{0.3}}
\put(4.058,4.992){\circle*{0.3}}
\put(1.718,4.006){\circle*{0.3}}
\put(2.394,3.){\circle*{0.3}}
\put(1.718,1.994){\circle*{0.3}}
\put(4.058,1.008){\circle*{0.3}}
\put(0,3){\vector(1,0){6}}
\put(3,0){\vector(0,1){6}}
\multiput(1,0.21)(0,0.2){29}{\circle*{0.07}}
\multiput(2,0.21)(0,0.2){29}{\circle*{0.07}}
\multiput(4,0.21)(0,0.2){29}{\circle*{0.07}}
\multiput(5,0.21)(0,0.2){29}{\circle*{0.07}}
\multiput(0.2,1.01)(0.2,0){29}{\circle*{0.07}}
\multiput(0.2,2.01)(0.2,0){29}{\circle*{0.07}}
\multiput(0.2,4.01)(0.2,0){29}{\circle*{0.07}}
\multiput(0.2,5.01)(0.2,0){29}{\circle*{0.07}}
\end{picture}
\begin{picture}(7,6)
\put(0.0,5.2){$3$}
\put(3.996,3.){\circle*{0.3}}
\put(4.004,4.994){\circle*{0.3}}
\put(2.766,3.){\circle*{0.3}}
\put(1.616,4.01){\circle*{0.3}}
\put(1.616,1.99){\circle*{0.3}}
\put(4.004,1.006){\circle*{0.3}}
\put(0,3){\vector(1,0){6}}
\put(3,0){\vector(0,1){6}}
\multiput(1,0.21)(0,0.2){29}{\circle*{0.07}}
\multiput(2,0.21)(0,0.2){29}{\circle*{0.07}}
\multiput(4,0.21)(0,0.2){29}{\circle*{0.07}}
\multiput(5,0.21)(0,0.2){29}{\circle*{0.07}}
\multiput(0.2,1.01)(0.2,0){29}{\circle*{0.07}}
\multiput(0.2,2.01)(0.2,0){29}{\circle*{0.07}}
\multiput(0.2,4.01)(0.2,0){29}{\circle*{0.07}}
\multiput(0.2,5.01)(0.2,0){29}{\circle*{0.07}}
\end{picture}
\begin{picture}(7,6)
\put(0.0,5.2){$4$}
\put(3.93,3.){\circle*{0.3}}
\put(3.938,4.996){\circle*{0.3}}
\put(1.562,4.012){\circle*{0.3}}
\put(1.562,1.988){\circle*{0.3}}
\put(3.07,3.){\circle*{0.3}}
\put(3.938,1.004){\circle*{0.3}}
\put(0,3){\vector(1,0){6}}
\put(3,0){\vector(0,1){6}}
\multiput(1,0.21)(0,0.2){29}{\circle*{0.07}}
\multiput(2,0.21)(0,0.2){29}{\circle*{0.07}}
\multiput(4,0.21)(0,0.2){29}{\circle*{0.07}}
\multiput(5,0.21)(0,0.2){29}{\circle*{0.07}}
\multiput(0.2,1.01)(0.2,0){29}{\circle*{0.07}}
\multiput(0.2,2.01)(0.2,0){29}{\circle*{0.07}}
\multiput(0.2,4.01)(0.2,0){29}{\circle*{0.07}}
\multiput(0.2,5.01)(0.2,0){29}{\circle*{0.07}}
\end{picture}
\end{center}

\begin{center}
\unitlength 12pt
\begin{picture}(7,6)
\put(0.0,5.2){$5$}
\put(3.848,3.){\circle*{0.3}}
\put(3.858,4.998){\circle*{0.3}}
\put(3.386,3.){\circle*{0.3}}
\put(1.524,4.012){\circle*{0.3}}
\put(1.524,1.988){\circle*{0.3}}
\put(3.858,1.002){\circle*{0.3}}
\put(0,3){\vector(1,0){6}}
\put(3,0){\vector(0,1){6}}
\multiput(1,0.21)(0,0.2){29}{\circle*{0.07}}
\multiput(2,0.21)(0,0.2){29}{\circle*{0.07}}
\multiput(4,0.21)(0,0.2){29}{\circle*{0.07}}
\multiput(5,0.21)(0,0.2){29}{\circle*{0.07}}
\multiput(0.2,1.01)(0.2,0){29}{\circle*{0.07}}
\multiput(0.2,2.01)(0.2,0){29}{\circle*{0.07}}
\multiput(0.2,4.01)(0.2,0){29}{\circle*{0.07}}
\multiput(0.2,5.01)(0.2,0){29}{\circle*{0.07}}
\end{picture}
\begin{picture}(7,6)
\put(0.0,5.2){$6^*$}
\put(3.768,3.038){\circle*{0.3}}
\put(3.734,4.982){\circle*{0.3}}
\put(1.496,4.014){\circle*{0.3}}
\put(1.496,1.986){\circle*{0.3}}
\put(3.768,2.962){\circle*{0.3}}
\put(3.734,1.018){\circle*{0.3}}
\put(0,3){\vector(1,0){6}}
\put(3,0){\vector(0,1){6}}
\multiput(1,0.21)(0,0.2){29}{\circle*{0.07}}
\multiput(2,0.21)(0,0.2){29}{\circle*{0.07}}
\multiput(4,0.21)(0,0.2){29}{\circle*{0.07}}
\multiput(5,0.21)(0,0.2){29}{\circle*{0.07}}
\multiput(0.2,1.01)(0.2,0){29}{\circle*{0.07}}
\multiput(0.2,2.01)(0.2,0){29}{\circle*{0.07}}
\multiput(0.2,4.01)(0.2,0){29}{\circle*{0.07}}
\multiput(0.2,5.01)(0.2,0){29}{\circle*{0.07}}
\end{picture}
\begin{picture}(7,6)
\put(0.0,5.2){$7$}
\put(4.656,3.){\circle*{0.3}}
\put(3.46,3.){\circle*{0.3}}
\put(3.46,4.998){\circle*{0.3}}
\put(1.48,4.014){\circle*{0.3}}
\put(1.48,1.986){\circle*{0.3}}
\put(3.46,1.002){\circle*{0.3}}
\put(0,3){\vector(1,0){6}}
\put(3,0){\vector(0,1){6}}
\multiput(1,0.21)(0,0.2){29}{\circle*{0.07}}
\multiput(2,0.21)(0,0.2){29}{\circle*{0.07}}
\multiput(4,0.21)(0,0.2){29}{\circle*{0.07}}
\multiput(5,0.21)(0,0.2){29}{\circle*{0.07}}
\multiput(0.2,1.01)(0.2,0){29}{\circle*{0.07}}
\multiput(0.2,2.01)(0.2,0){29}{\circle*{0.07}}
\multiput(0.2,4.01)(0.2,0){29}{\circle*{0.07}}
\multiput(0.2,5.01)(0.2,0){29}{\circle*{0.07}}
\end{picture}
\end{center}

\paragraph{Group 2.}
\begin{center}
\unitlength 12pt
\begin{picture}(7,6)
\put(0.0,5.2){$8$}
\put(3.412,3.){\circle*{0.3}}
\put(3.21,3.998){\circle*{0.3}}
\put(3.41,5.){\circle*{0.3}}
\put(1.34,3.){\circle*{0.3}}
\put(3.21,2.002){\circle*{0.3}}
\put(3.41,1.){\circle*{0.3}}
\put(0,3){\vector(1,0){6}}
\put(3,0){\vector(0,1){6}}
\multiput(1,0.21)(0,0.2){29}{\circle*{0.07}}
\multiput(2,0.21)(0,0.2){29}{\circle*{0.07}}
\multiput(4,0.21)(0,0.2){29}{\circle*{0.07}}
\multiput(5,0.21)(0,0.2){29}{\circle*{0.07}}
\multiput(0.2,1.01)(0.2,0){29}{\circle*{0.07}}
\multiput(0.2,2.01)(0.2,0){29}{\circle*{0.07}}
\multiput(0.2,4.01)(0.2,0){29}{\circle*{0.07}}
\multiput(0.2,5.01)(0.2,0){29}{\circle*{0.07}}
\end{picture}
\begin{picture}(7,6)
\put(0.0,5.2){$9$}
\put(3.216,3.998){\circle*{0.3}}
\put(3.12,4.998){\circle*{0.3}}
\put(2.208,3.){\circle*{0.3}}
\put(3.118,3.){\circle*{0.3}}
\put(3.12,1.002){\circle*{0.3}}
\put(3.216,2.002){\circle*{0.3}}
\put(0,3){\vector(1,0){6}}
\put(3,0){\vector(0,1){6}}
\multiput(1,0.21)(0,0.2){29}{\circle*{0.07}}
\multiput(2,0.21)(0,0.2){29}{\circle*{0.07}}
\multiput(4,0.21)(0,0.2){29}{\circle*{0.07}}
\multiput(5,0.21)(0,0.2){29}{\circle*{0.07}}
\multiput(0.2,1.01)(0.2,0){29}{\circle*{0.07}}
\multiput(0.2,2.01)(0.2,0){29}{\circle*{0.07}}
\multiput(0.2,4.01)(0.2,0){29}{\circle*{0.07}}
\multiput(0.2,5.01)(0.2,0){29}{\circle*{0.07}}
\end{picture}
\begin{picture}(7,6)
\put(0.0,5.2){$10$}
\put(3.138,3.998){\circle*{0.3}}
\put(3.024,4.998){\circle*{0.3}}
\put(2.652,3.){\circle*{0.3}}
\put(3.02,3.){\circle*{0.3}}
\put(3.024,1.002){\circle*{0.3}}
\put(3.138,2.002){\circle*{0.3}}
\put(0,3){\vector(1,0){6}}
\put(3,0){\vector(0,1){6}}
\multiput(1,0.21)(0,0.2){29}{\circle*{0.07}}
\multiput(2,0.21)(0,0.2){29}{\circle*{0.07}}
\multiput(4,0.21)(0,0.2){29}{\circle*{0.07}}
\multiput(5,0.21)(0,0.2){29}{\circle*{0.07}}
\multiput(0.2,1.01)(0.2,0){29}{\circle*{0.07}}
\multiput(0.2,2.01)(0.2,0){29}{\circle*{0.07}}
\multiput(0.2,4.01)(0.2,0){29}{\circle*{0.07}}
\multiput(0.2,5.01)(0.2,0){29}{\circle*{0.07}}
\end{picture}
\begin{picture}(7,6)
\put(0.0,5.2){$11^*$}
\put(3.,2.){\circle*{0.3}}
\put(3.,4.){\circle*{0.3}}
\put(3.,3.036){\circle*{0.3}}
\put(3.,4.986){\circle*{0.3}}
\put(3.,2.964){\circle*{0.3}}
\put(3.,1.014){\circle*{0.3}}
\put(0,3){\vector(1,0){6}}
\put(3,0){\vector(0,1){6}}
\multiput(1,0.21)(0,0.2){29}{\circle*{0.07}}
\multiput(2,0.21)(0,0.2){29}{\circle*{0.07}}
\multiput(4,0.21)(0,0.2){29}{\circle*{0.07}}
\multiput(5,0.21)(0,0.2){29}{\circle*{0.07}}
\multiput(0.2,1.01)(0.2,0){29}{\circle*{0.07}}
\multiput(0.2,2.01)(0.2,0){29}{\circle*{0.07}}
\multiput(0.2,4.01)(0.2,0){29}{\circle*{0.07}}
\multiput(0.2,5.01)(0.2,0){29}{\circle*{0.07}}
\end{picture}
\end{center}

\begin{center}
\unitlength 12pt
\begin{picture}(7,6)
\put(0.0,5.2){$12$}
\put(2.98,3.){\circle*{0.3}}
\put(2.976,4.998){\circle*{0.3}}
\put(2.862,3.998){\circle*{0.3}}
\put(2.862,2.002){\circle*{0.3}}
\put(2.976,1.002){\circle*{0.3}}
\put(3.348,3.){\circle*{0.3}}
\put(0,3){\vector(1,0){6}}
\put(3,0){\vector(0,1){6}}
\multiput(1,0.21)(0,0.2){29}{\circle*{0.07}}
\multiput(2,0.21)(0,0.2){29}{\circle*{0.07}}
\multiput(4,0.21)(0,0.2){29}{\circle*{0.07}}
\multiput(5,0.21)(0,0.2){29}{\circle*{0.07}}
\multiput(0.2,1.01)(0.2,0){29}{\circle*{0.07}}
\multiput(0.2,2.01)(0.2,0){29}{\circle*{0.07}}
\multiput(0.2,4.01)(0.2,0){29}{\circle*{0.07}}
\multiput(0.2,5.01)(0.2,0){29}{\circle*{0.07}}
\end{picture}
\begin{picture}(7,6)
\put(0.0,5.2){$13$}
\put(2.882,3.){\circle*{0.3}}
\put(2.88,4.998){\circle*{0.3}}
\put(2.784,3.998){\circle*{0.3}}
\put(2.784,2.002){\circle*{0.3}}
\put(2.88,1.002){\circle*{0.3}}
\put(3.792,3.){\circle*{0.3}}
\put(0,3){\vector(1,0){6}}
\put(3,0){\vector(0,1){6}}
\multiput(1,0.21)(0,0.2){29}{\circle*{0.07}}
\multiput(2,0.21)(0,0.2){29}{\circle*{0.07}}
\multiput(4,0.21)(0,0.2){29}{\circle*{0.07}}
\multiput(5,0.21)(0,0.2){29}{\circle*{0.07}}
\multiput(0.2,1.01)(0.2,0){29}{\circle*{0.07}}
\multiput(0.2,2.01)(0.2,0){29}{\circle*{0.07}}
\multiput(0.2,4.01)(0.2,0){29}{\circle*{0.07}}
\multiput(0.2,5.01)(0.2,0){29}{\circle*{0.07}}
\end{picture}
\begin{picture}(7,6)
\put(0.0,5.2){$14$}
\put(4.66,3.){\circle*{0.3}}
\put(2.79,3.998){\circle*{0.3}}
\put(2.59,5.){\circle*{0.3}}
\put(2.588,3.){\circle*{0.3}}
\put(2.79,2.002){\circle*{0.3}}
\put(2.59,1.){\circle*{0.3}}
\put(0,3){\vector(1,0){6}}
\put(3,0){\vector(0,1){6}}
\multiput(1,0.21)(0,0.2){29}{\circle*{0.07}}
\multiput(2,0.21)(0,0.2){29}{\circle*{0.07}}
\multiput(4,0.21)(0,0.2){29}{\circle*{0.07}}
\multiput(5,0.21)(0,0.2){29}{\circle*{0.07}}
\multiput(0.2,1.01)(0.2,0){29}{\circle*{0.07}}
\multiput(0.2,2.01)(0.2,0){29}{\circle*{0.07}}
\multiput(0.2,4.01)(0.2,0){29}{\circle*{0.07}}
\multiput(0.2,5.01)(0.2,0){29}{\circle*{0.07}}
\end{picture}
\end{center}

\paragraph{Group 3.}
\begin{center}
\unitlength 12pt
\begin{picture}(7,6)
\put(0.0,5.2){$15$}
\put(4.52,1.986){\circle*{0.3}}
\put(4.52,4.014){\circle*{0.3}}
\put(2.54,4.998){\circle*{0.3}}
\put(1.344,3.){\circle*{0.3}}
\put(2.54,3.){\circle*{0.3}}
\put(2.54,1.002){\circle*{0.3}}
\put(0,3){\vector(1,0){6}}
\put(3,0){\vector(0,1){6}}
\multiput(1,0.21)(0,0.2){29}{\circle*{0.07}}
\multiput(2,0.21)(0,0.2){29}{\circle*{0.07}}
\multiput(4,0.21)(0,0.2){29}{\circle*{0.07}}
\multiput(5,0.21)(0,0.2){29}{\circle*{0.07}}
\multiput(0.2,1.01)(0.2,0){29}{\circle*{0.07}}
\multiput(0.2,2.01)(0.2,0){29}{\circle*{0.07}}
\multiput(0.2,4.01)(0.2,0){29}{\circle*{0.07}}
\multiput(0.2,5.01)(0.2,0){29}{\circle*{0.07}}
\end{picture}
\begin{picture}(7,6)
\put(0.0,5.2){$16^*$}
\put(4.504,4.014){\circle*{0.3}}
\put(2.232,3.038){\circle*{0.3}}
\put(2.266,4.982){\circle*{0.3}}
\put(2.232,2.962){\circle*{0.3}}
\put(2.266,1.018){\circle*{0.3}}
\put(4.504,1.986){\circle*{0.3}}
\put(0,3){\vector(1,0){6}}
\put(3,0){\vector(0,1){6}}
\multiput(1,0.21)(0,0.2){29}{\circle*{0.07}}
\multiput(2,0.21)(0,0.2){29}{\circle*{0.07}}
\multiput(4,0.21)(0,0.2){29}{\circle*{0.07}}
\multiput(5,0.21)(0,0.2){29}{\circle*{0.07}}
\multiput(0.2,1.01)(0.2,0){29}{\circle*{0.07}}
\multiput(0.2,2.01)(0.2,0){29}{\circle*{0.07}}
\multiput(0.2,4.01)(0.2,0){29}{\circle*{0.07}}
\multiput(0.2,5.01)(0.2,0){29}{\circle*{0.07}}
\end{picture}
\begin{picture}(7,6)
\put(0.0,5.2){$17$}
\put(4.476,4.012){\circle*{0.3}}
\put(2.142,4.998){\circle*{0.3}}
\put(2.152,3.){\circle*{0.3}}
\put(2.142,1.002){\circle*{0.3}}
\put(2.614,3.){\circle*{0.3}}
\put(4.476,1.988){\circle*{0.3}}
\put(0,3){\vector(1,0){6}}
\put(3,0){\vector(0,1){6}}
\multiput(1,0.21)(0,0.2){29}{\circle*{0.07}}
\multiput(2,0.21)(0,0.2){29}{\circle*{0.07}}
\multiput(4,0.21)(0,0.2){29}{\circle*{0.07}}
\multiput(5,0.21)(0,0.2){29}{\circle*{0.07}}
\multiput(0.2,1.01)(0.2,0){29}{\circle*{0.07}}
\multiput(0.2,2.01)(0.2,0){29}{\circle*{0.07}}
\multiput(0.2,4.01)(0.2,0){29}{\circle*{0.07}}
\multiput(0.2,5.01)(0.2,0){29}{\circle*{0.07}}
\end{picture}
\begin{picture}(7,6)
\put(0.0,5.2){$18$}
\put(4.438,4.012){\circle*{0.3}}
\put(2.93,3.){\circle*{0.3}}
\put(2.062,4.996){\circle*{0.3}}
\put(2.07,3.){\circle*{0.3}}
\put(2.062,1.004){\circle*{0.3}}
\put(4.438,1.988){\circle*{0.3}}
\put(0,3){\vector(1,0){6}}
\put(3,0){\vector(0,1){6}}
\multiput(1,0.21)(0,0.2){29}{\circle*{0.07}}
\multiput(2,0.21)(0,0.2){29}{\circle*{0.07}}
\multiput(4,0.21)(0,0.2){29}{\circle*{0.07}}
\multiput(5,0.21)(0,0.2){29}{\circle*{0.07}}
\multiput(0.2,1.01)(0.2,0){29}{\circle*{0.07}}
\multiput(0.2,2.01)(0.2,0){29}{\circle*{0.07}}
\multiput(0.2,4.01)(0.2,0){29}{\circle*{0.07}}
\multiput(0.2,5.01)(0.2,0){29}{\circle*{0.07}}
\end{picture}
\end{center}

\begin{center}
\unitlength 12pt
\begin{picture}(7,6)
\put(0.0,5.2){$19$}
\put(4.384,4.01){\circle*{0.3}}
\put(1.996,4.994){\circle*{0.3}}
\put(2.004,3.){\circle*{0.3}}
\put(1.996,1.006){\circle*{0.3}}
\put(3.234,3.){\circle*{0.3}}
\put(4.384,1.99){\circle*{0.3}}
\put(0,3){\vector(1,0){6}}
\put(3,0){\vector(0,1){6}}
\multiput(1,0.21)(0,0.2){29}{\circle*{0.07}}
\multiput(2,0.21)(0,0.2){29}{\circle*{0.07}}
\multiput(4,0.21)(0,0.2){29}{\circle*{0.07}}
\multiput(5,0.21)(0,0.2){29}{\circle*{0.07}}
\multiput(0.2,1.01)(0.2,0){29}{\circle*{0.07}}
\multiput(0.2,2.01)(0.2,0){29}{\circle*{0.07}}
\multiput(0.2,4.01)(0.2,0){29}{\circle*{0.07}}
\multiput(0.2,5.01)(0.2,0){29}{\circle*{0.07}}
\end{picture}
\begin{picture}(7,6)
\put(0.0,5.2){$20$}
\put(3.606,3.){\circle*{0.3}}
\put(4.282,4.006){\circle*{0.3}}
\put(1.942,4.992){\circle*{0.3}}
\put(1.946,3.){\circle*{0.3}}
\put(1.942,1.008){\circle*{0.3}}
\put(4.282,1.994){\circle*{0.3}}
\put(0,3){\vector(1,0){6}}
\put(3,0){\vector(0,1){6}}
\multiput(1,0.21)(0,0.2){29}{\circle*{0.07}}
\multiput(2,0.21)(0,0.2){29}{\circle*{0.07}}
\multiput(4,0.21)(0,0.2){29}{\circle*{0.07}}
\multiput(5,0.21)(0,0.2){29}{\circle*{0.07}}
\multiput(0.2,1.01)(0.2,0){29}{\circle*{0.07}}
\multiput(0.2,2.01)(0.2,0){29}{\circle*{0.07}}
\multiput(0.2,4.01)(0.2,0){29}{\circle*{0.07}}
\multiput(0.2,5.01)(0.2,0){29}{\circle*{0.07}}
\end{picture}
\begin{picture}(7,6)
\put(0.0,5.2){$21$}
\put(4.43,3.){\circle*{0.3}}
\put(3.916,4.){\circle*{0.3}}
\put(1.912,4.992){\circle*{0.3}}
\put(1.916,3.){\circle*{0.3}}
\put(1.912,1.008){\circle*{0.3}}
\put(3.916,2.){\circle*{0.3}}
\put(0,3){\vector(1,0){6}}
\put(3,0){\vector(0,1){6}}
\multiput(1,0.21)(0,0.2){29}{\circle*{0.07}}
\multiput(2,0.21)(0,0.2){29}{\circle*{0.07}}
\multiput(4,0.21)(0,0.2){29}{\circle*{0.07}}
\multiput(5,0.21)(0,0.2){29}{\circle*{0.07}}
\multiput(0.2,1.01)(0.2,0){29}{\circle*{0.07}}
\multiput(0.2,2.01)(0.2,0){29}{\circle*{0.07}}
\multiput(0.2,4.01)(0.2,0){29}{\circle*{0.07}}
\multiput(0.2,5.01)(0.2,0){29}{\circle*{0.07}}
\end{picture}
\end{center}

Our basic principle is to assign larger riggings for right strings
and smaller riggings for left strings.
Then if we compare the positions of 2-strings in Groups 1, 2 and 3,
we can assign $r_1=0$ for Group 1, $r_1=1$ for Group 2
and $r_1=2$ for Group 3 without ambiguity.
Then, within each group, we can assign $r_2=0,1,\ldots,6$
from smaller solutions number to larger solutions number.
For example, solution $\#6^*$ corresponds to $(0,5)$,
solution $\#11^*$ corresponds to $(1,3)$ and solution $\#16^*$ corresponds to $(2,1)$, respectively.
This is in agreement with the results of \cite{DG} based on the Bethe--Takahashi quantum number.

We observe that within each group, the positions of 2-string and 3-string are almost unchanged
whereas 1-string moves from left to right.
In this interpretation, we see that there are collisions of
1-string and 3-string at the solutions $\#6^*$, $\#11^*$ and $\#16^*$.
To summarize, we have demonstrated that the non self-conjugate strings of the form
(\ref{eq:3+1_non_self-conjugate-strings}) is generated by a fusion of
a 3-string and a 1-string even if they look irregular.
As we can see, they fit naturally with the rigged configurations picture.

Finally we remark that there are small deviations for the real parts of three roots within each 3-string.
However, as we can see in the following numerical tables, there are always clear
gaps between real parts of neighboring solutions.
Therefore, in this case, there is no ambiguity to introduce an order based on the real parts of the 3-strings.
In the following tables, the first column indicates the solutions number.
Within each solution, the first line contains 3-string and the second line contains
2-string and 1-string.

\paragraph{Group 1.}
\[
\begin{array}{l|lll}
\hline\hline
1&0.54241927&0.54455699+0.99639165 i&0.54455699-0.99639165 i\\
&-0.45810568+0.50017785 i&-0.45810568-0.50017785 i&-0.71532188\\
\hline
2&0.52708058&0.52957875+0.99660493 i&0.52957875-0.99660493 i\\
&-0.64127830+0.50335013 i&-0.64127830-0.50335013 i&-0.30368149\\
\hline
3&0.49893578&0.50200196+0.99724969 i&0.50200196-0.99724969 i\\
&-0.69284388+0.50515234 i&-0.69284388-0.50515234 i&-0.11725196\\
\hline
4&0.46564665&0.46941736+0.99809739 i&0.46941736-0.99809739 i\\
&-0.71976299+0.50614240 i&-0.71976299-0.50614240 i&0.035044606\\
\hline
5&0.42430010&0.42960641+0.99941330 i&0.42960641-0.99941330 i\\
&-0.73869344+0.50683156 i&-0.73869344-0.50683156 i&0.19387395\\
\hline
6^*&0.38490522+0.01906127 i&0.36730804+0.99179719 i&0.36730804-0.99179719 i\\
&-0.75221326+0.50729383 i&-0.75221326-0.50729383 i&0.38490522-0.01906127 i\\
\hline
7&0.23056669&0.23083274+0.99967059 i&0.23083274-0.99967059 i\\
&-0.76056174+0.50745313 i&-0.76056174-0.50745313 i&0.82889133\\
\hline
\end{array}
\]

\paragraph{Group 2.}
\[
\begin{array}{l|lll}
\hline\hline
8&0.20669577&0.20597572+1.00038608 i&0.20597572-1.00038608 i\\
&0.10578435+0.50000000 i&0.10578435-0.50000000 i&-0.83021590\\
\hline
9&0.059726272&0.06007063+0.99927337 i&0.06007063-0.99927337 i\\
&0.10847310+0.50000000 i&0.10847310-0.50000000 i&-0.39681373\\
\hline
10&0.010757119&0.01249979+0.99958901 i&0.01249979-0.99958901 i\\
&0.06941354+0.50000000 i&0.06941354-0.50000000 i&-0.17458378\\
\hline
11^*&0.018539900 i&0.99377501 i&-0.99377501 i\\
&0.50000000 i&-0.50000000 i&-0.018539900 i\\
\hline
12&-0.010757119&-0.01249979+0.99958901 i&-0.01249979-0.99958901 i\\
&-0.06941354+0.50000000 i&-0.06941354-0.50000000 i&0.17458378\\
\hline
13&-0.059726272&-0.06007063+0.99927337 i&-0.06007063-0.99927337 i\\
&-0.10847310+0.50000000 i&-0.10847310-0.50000000 i&0.39681373\\
\hline
14&-0.20669577&-0.20597572+1.00038608 i&-0.20597572-1.00038608 i\\
&-0.10578435+0.50000000 i&-0.10578435-0.50000000 i&0.83021590\\
\hline
\end{array}
\]

\paragraph{Group 3.}
\[
\begin{array}{l|lll}
\hline\hline
15&-0.23056669&-0.23083274+0.99967059 i&-0.23083274-0.99967059 i\\
&0.76056174+0.50745313 i&0.76056174-0.50745313 i&-0.82889133\\
\hline
16^*&-0.38490522+0.01906127 i&-0.36730804+0.99179719 i&-0.36730804-0.99179719i\\
&0.75221326+0.50729383 i&0.75221326-0.50729383 i&-0.38490522-0.01906127 i\\
\hline
17&-0.42430010&-0.42960641+0.99941330 i&-0.42960641-0.99941330 i\\
&0.73869344+0.50683156 i&0.73869344-0.50683156 i&-0.19387395\\
\hline
18&-0.46564665&-0.46941736+0.99809739 i&-0.46941736-0.99809739 i\\
&0.71976299+0.50614240 i&0.71976299-0.50614240 i&-0.035044606\\
\hline
19&-0.49893578&-0.50200196+0.99724969 i&-0.50200196-0.99724969 i\\
&0.69284388+0.50515234 i&0.69284388-0.50515234 i&0.11725196\\
\hline
20&-0.52708058&-0.52957875+0.99660493 i&-0.52957875-0.99660493 i\\
&0.64127830+0.50335013 i&0.64127830-0.50335013 i&0.30368149\\
\hline
21&-0.54241927&-0.54455699+0.99639165 i&-0.54455699-0.99639165 i\\
&0.45810568+0.50017785 i&0.45810568-0.50017785 i&0.71532188\\
\hline
\end{array}
\]

In \cite{DG} another set of non self-conjugate strings are reported at $N=12$ and $\nu=(3,1,1)$.
The total number of the corresponding solutions is 84.
Let us denote the solutions in the following way:
\[
\{\lambda_1,\lambda_2,
\eta^{(i)},\eta^{(0)},\eta^{(-i)}\}.
\]
Here $\lambda_1$ and $\lambda_2$ are 1-strings ($\lambda_1<\lambda_2$)
and $\eta$ corresponds to the 3-string
($\mathrm{Im}(\eta^{(i)})>\mathrm{Im}(\eta^{(0)})>\mathrm{Im}(\eta^{(-i)})$).
Below we denote the corresponding rigged configurations
$\{(3,r_1),(1,r_2),(1,r_3)\}$ by $(r_1,r_2,r_3)$.
Note that it is enough to consider the case $\mathrm{Re}(\eta)\geq 0$
since we can derive the remaining rigged configurations by the flip operation
(see \cite[Section 3]{KiSa14}).
We arrange the corresponding 84 solutions according to the descending order of
$\mathrm{Re}(\eta^{(\pm i)})$.\footnote{Alternatively we may define the order
according to $\mathrm{Re}(\eta^{(0)})$. This choice does not affect the final result.
However we prefer to use $\mathrm{Re}(\eta^{(\pm i)})$ since our interest here
is a special behavior of $\eta^{(0)}$.}
Then we recognize a regular structure related with the rigged configurations.
To be more specific, we start from the first solution and proceed the list one by one.
We can identify $(2,0,0)$, $(2,0,1)$, $\cdots$, $(2,0,6)$
by choosing the solutions such that $\lambda_1$ is decreasing and $\lambda_2$ is increasing.
After removing the selected seven solutions from the list,
we determine the solutions corresponding to $(2,1,1)$, $(2,1,2)$, $\cdots$, $(2,1,6)$ similarly.
We repeat the procedure until $(2,6,6)$.
Note that $\lambda_1$ for $(2,i,i)$ should be increasing.
Next we start from $(1,0,0)$ and proceed similarly.
The resulting rigged configurations coincide with the identification of \cite{GD2}
in terms of the Bethe--Takahashi quantum number.
As the result, we can see that the non self-conjugate strings in this case
is also obtained by a fusion of a 1-string and a 3-string as in the previous case.

\section{Exceptional real solutions}
\label{sec:korepin}
In this section, we analyze the famous counter example to the string hypothesis
discovered by Essler, Korepin and Schoutens \cite{EKS:1992} in 1992.
For this purpose we consider the case $N=25$ and $\ell=2$
where the corresponding Bethe ansatz equations are
\begin{align}
\left(\frac{\lambda_1+\frac{i}{2}}{\lambda_1-\frac{i}{2}}\right)^{25}
=\frac{\lambda_1-\lambda_2+i}{\lambda_1-\lambda_2-i},
\qquad
\left(\frac{\lambda_2+\frac{i}{2}}{\lambda_2-\frac{i}{2}}\right)^{25}
=\frac{\lambda_2-\lambda_1+i}{\lambda_2-\lambda_1-i}.
\end{align}
We consider the solutions which satisfy $\lambda_1\neq\lambda_2$.
Moreover we identify the solutions which are connected by the permutation
such as $\{\lambda_1,\lambda_2\}$ and $\{\lambda_2,\lambda_1\}$.
Then we have 255 real solutions and 21 complex solutions.

We start from the analysis of the real solutions.
In the rigged configurations picture, they should correspond to the following rigged configurations:
\begin{center}
\unitlength 12pt
\begin{picture}(3,2)
\put(0,1.1){21}
\put(0,0.1){21}
\put(1.2,0){\Yboxdim12pt\yng(1,1)}
\put(2.5,1.1){$r_1$}
\put(2.5,0.1){$r_2$}
\end{picture}
\end{center}
During this section, we refer to the above rigged configuration as $(r_1,r_2)$.
Then we have $21\geq r_1\geq r_2\geq 0$.
Now we introduce an order for the set of the real solutions $\{\lambda_1,\lambda_2\}$.
Without loss of generality we can suppose that $\lambda_1<\lambda_2$.
Then we arrange 255 real solutions in the descending order of the larger roots $\lambda_2$.
We call them $\{\lambda_1^{(1)},\lambda_2^{(1)}\}$, $\{\lambda_1^{(2)},\lambda_2^{(2)}\}$,
$\cdots$, $\{\lambda_1^{(255)},\lambda_2^{(255)}\}$.
By definition, we have $\lambda_2^{(i)}>\lambda_2^{(i+1)}$.
Remarkably, this simple procedure clarifies a natural relationship between the real solutions
and the rigged configurations and eventually two exceptional real solutions appear as the
elements which violate the rigged configurations structure.

Let us analyze the real solutions starting from $\{\lambda_1^{(1)},\lambda_2^{(1)}\}$,
$\{\lambda_1^{(2)},\lambda_2^{(2)}\}$, $\cdots$.
In the left of the following two diagrams, we show $\{\lambda_1^{(1)},\lambda_2^{(1)}\}$,
$\{\lambda_1^{(2)},\lambda_2^{(2)}\}$, $\cdots$, $\{\lambda_1^{(22)},\lambda_2^{(22)}\}$
from top to bottom.
The right diagram contains the solution $\{\lambda_1^{(23)},\lambda_2^{(23)}\}$.
In both diagrams, the rightward arrows represent the real axis on which we locate
each solution $\{\lambda_1^{(i)},\lambda_2^{(i)}\}$.
Here the spacing of the dotted lines is $0.5$ and the vertical arrows indicate the position of the origin 0.
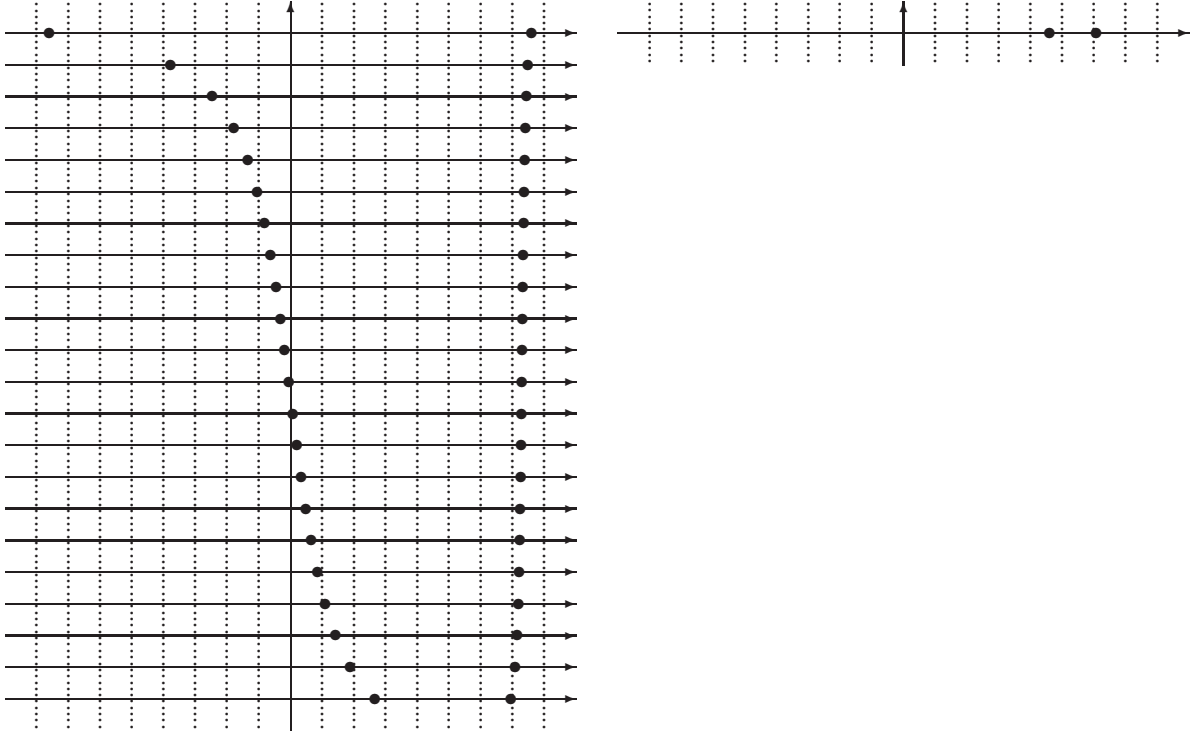
\begin{center}
\unitlength 12pt
\begin{picture}(19,23)
\put(9,0){\vector(0,1){23}}
\multiput(0,1)(0,1){22}{\vector(1,0){18}}
\put(16.60,22){\circle*{0.3}}
\put(1.40,22){\circle*{0.3}}
\put(16.49,21){\circle*{0.3}}
\put(5.22,21){\circle*{0.3}}
\put(16.44,20){\circle*{0.3}}
\put(6.53,20){\circle*{0.3}}
\put(16.40,19){\circle*{0.3}}
\put(7.22,19){\circle*{0.3}}
\put(16.38,18){\circle*{0.3}}
\put(7.65,18){\circle*{0.3}}
\put(16.36,17){\circle*{0.3}}
\put(7.95,17){\circle*{0.3}}
\put(16.35,16){\circle*{0.3}}
\put(8.189,16){\circle*{0.3}}
\put(16.33,15){\circle*{0.3}}
\put(8.379,15){\circle*{0.3}}
\put(16.32,14){\circle*{0.3}}
\put(8.542,14){\circle*{0.3}}
\put(16.31,13){\circle*{0.3}}
\put(8.686,13){\circle*{0.3}}
\put(16.30,12){\circle*{0.3}}
\put(8.820,12){\circle*{0.3}}
\put(16.29,11){\circle*{0.3}}
\put(8.9478,11){\circle*{0.3}}
\put(16.28,10){\circle*{0.3}}
\put(9.0738,10){\circle*{0.3}}
\put(16.27,9){\circle*{0.3}}
\put(9.202,9){\circle*{0.3}}
\put(16.26,8){\circle*{0.3}}
\put(9.337,8){\circle*{0.3}}
\put(16.24,7){\circle*{0.3}}
\put(9.485,7){\circle*{0.3}}
\put(16.23,6){\circle*{0.3}}
\put(9.651,6){\circle*{0.3}}
\put(16.21,5){\circle*{0.3}}
\put(9.848,5){\circle*{0.3}}
\put(16.18,4){\circle*{0.3}}
\put(10.09,4){\circle*{0.3}}
\put(16.14,3){\circle*{0.3}}
\put(10.42,3){\circle*{0.3}}
\put(16.08,2){\circle*{0.3}}
\put(10.88,2){\circle*{0.3}}
\put(15.94,1){\circle*{0.3}}
\put(11.66,1){\circle*{0.3}}
\multiput(1,0.12)(1,0){8}{
\multiput(0,0)(0,0.2){115}{\circle*{0.07}}
}
\multiput(10,0.12)(1,0){8}{
\multiput(0,0)(0,0.2){115}{\circle*{0.07}}
}
\end{picture}
\begin{picture}(18,23)(0,-21)
\put(9,0){\vector(0,1){2}}
\multiput(0,1)(0,1){1}{\vector(1,0){18}}
\put(15.08,1){\circle*{0.3}}
\put(13.60,1){\circle*{0.3}}
\multiput(1,0.12)(1,0){8}{
\multiput(0,0)(0,0.2){10}{\circle*{0.07}}
}
\multiput(10,0.12)(1,0){8}{
\multiput(0,0)(0,0.2){10}{\circle*{0.07}}
}
\end{picture}
\end{center}
As we can see in the left one of the above diagrams, we can confirm the following property:
\begin{align*}
\lambda_2^{(1)}>\lambda_2^{(2)}>\cdots >\lambda_2^{(22)},\\
\lambda_1^{(1)}<\lambda_1^{(2)}<\cdots <\lambda_1^{(22)}.
\end{align*}
Again, we correspond the larger riggings with the right roots and smaller riggings with the left roots.
Therefore we have the following correspondence between the solutions
$\{\lambda_1^{(i)},\lambda_2^{(i)}\}$ and the rigged configurations $(r_1,r_2)$:
\begin{align*}
\{\lambda_1^{(1)},\lambda_2^{(1)}\}\longmapsto (21,0),\,\,
\{\lambda_1^{(2)},\lambda_2^{(2)}\}\longmapsto (21,1),\cdots,
\{\lambda_1^{(22)},\lambda_2^{(22)}\}\longmapsto (21,21).
\end{align*}

Let us postpone the analysis of the solution $\{\lambda_1^{(23)},\lambda_2^{(23)}\}$ for a while
and consider the analysis of the following two groups of the solutions:
\begin{align*}
&\{\lambda_1^{(24)},\lambda_2^{(24)}\}, \{\lambda_1^{(25)},\lambda_2^{(25)}\}, \cdots,
\{\lambda_1^{(44)},\lambda_2^{(44)}\},\\
&\{\lambda_1^{(45)},\lambda_2^{(45)}\}, \{\lambda_1^{(46)},\lambda_2^{(46)}\}, \cdots,
\{\lambda_1^{(64)},\lambda_2^{(64)}\}.
\end{align*}
The first group of the solutions are presented in the left one of the following two diagrams
and the second group in the right diagram.
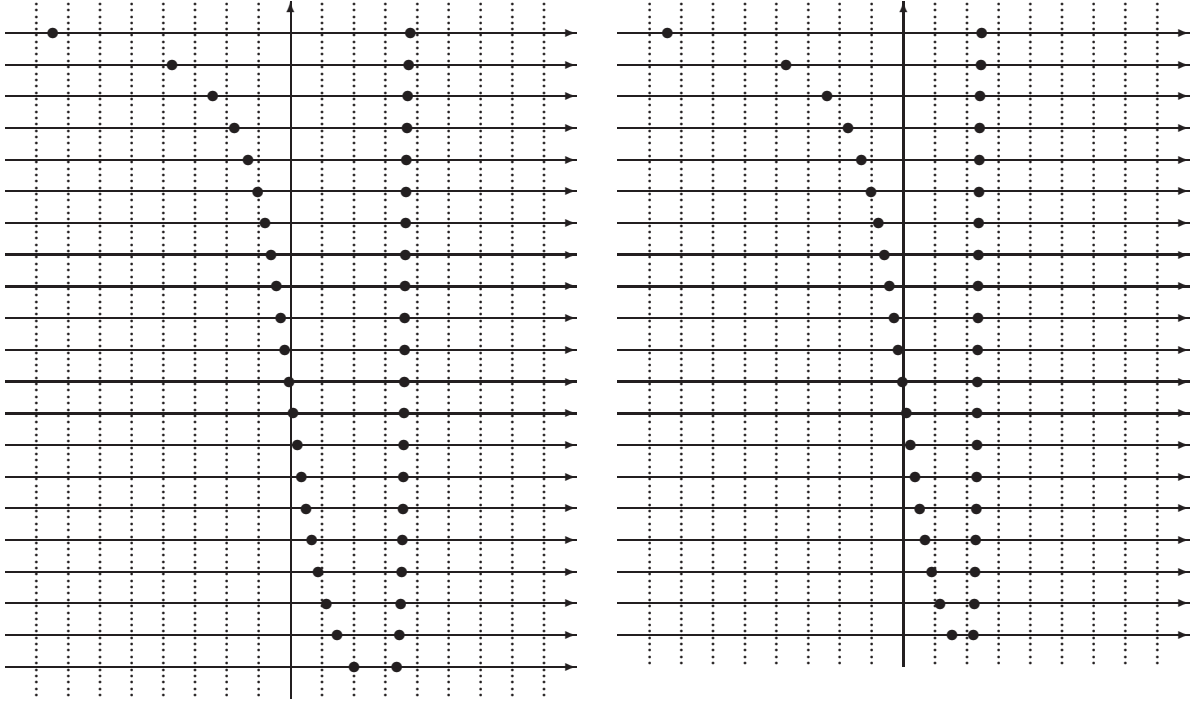
\begin{center}
\unitlength 12pt
\begin{picture}(19,22)
\put(9,0){\vector(0,1){22}}
\multiput(0,1)(0,1){21}{\vector(1,0){18}}
\put(12.78,21){\circle*{0.3}}
\put(1.51,21){\circle*{0.3}}
\put(12.73,20){\circle*{0.3}}
\put(5.27,20){\circle*{0.3}}
\put(12.70,19){\circle*{0.3}}
\put(6.56,19){\circle*{0.3}}
\put(12.68,18){\circle*{0.3}}
\put(7.24,18){\circle*{0.3}}
\put(12.66,17){\circle*{0.3}}
\put(7.67,17){\circle*{0.3}}
\put(12.65,16){\circle*{0.3}}
\put(7.97,16){\circle*{0.3}}
\put(12.64,15){\circle*{0.3}}
\put(8.201,15){\circle*{0.3}}
\put(12.63,14){\circle*{0.3}}
\put(8.390,14){\circle*{0.3}}
\put(12.62,13){\circle*{0.3}}
\put(8.552,13){\circle*{0.3}}
\put(12.61,12){\circle*{0.3}}
\put(8.696,12){\circle*{0.3}}
\put(12.60,11){\circle*{0.3}}
\put(8.829,11){\circle*{0.3}}
\put(12.59,10){\circle*{0.3}}
\put(8.9573,10){\circle*{0.3}}
\put(12.58,9){\circle*{0.3}}
\put(9.0838,9){\circle*{0.3}}
\put(12.57,8){\circle*{0.3}}
\put(9.213,8){\circle*{0.3}}
\put(12.56,7){\circle*{0.3}}
\put(9.350,7){\circle*{0.3}}
\put(12.55,6){\circle*{0.3}}
\put(9.499,6){\circle*{0.3}}
\put(12.53,5){\circle*{0.3}}
\put(9.669,5){\circle*{0.3}}
\put(12.51,4){\circle*{0.3}}
\put(9.872,4){\circle*{0.3}}
\put(12.48,3){\circle*{0.3}}
\put(10.13,3){\circle*{0.3}}
\put(12.44,2){\circle*{0.3}}
\put(10.47,2){\circle*{0.3}}
\put(12.35,1){\circle*{0.3}}
\put(11.00,1){\circle*{0.3}}
\multiput(1,0.12)(1,0){8}{
\multiput(0,0)(0,0.2){110}{\circle*{0.07}}
}
\multiput(10,0.12)(1,0){8}{
\multiput(0,0)(0,0.2){110}{\circle*{0.07}}
}
\end{picture}
\begin{picture}(18,22)(0,-1)
\put(9,0){\vector(0,1){21}}
\multiput(0,1)(0,1){20}{\vector(1,0){18}}
\put(11.47,20){\circle*{0.3}}
\put(1.56,20){\circle*{0.3}}
\put(11.44,19){\circle*{0.3}}
\put(5.30,19){\circle*{0.3}}
\put(11.41,18){\circle*{0.3}}
\put(6.59,18){\circle*{0.3}}
\put(11.40,17){\circle*{0.3}}
\put(7.26,17){\circle*{0.3}}
\put(11.39,16){\circle*{0.3}}
\put(7.68,16){\circle*{0.3}}
\put(11.38,15){\circle*{0.3}}
\put(7.98,15){\circle*{0.3}}
\put(11.37,14){\circle*{0.3}}
\put(8.210,14){\circle*{0.3}}
\put(11.36,13){\circle*{0.3}}
\put(8.398,13){\circle*{0.3}}
\put(11.35,12){\circle*{0.3}}
\put(8.559,12){\circle*{0.3}}
\put(11.35,11){\circle*{0.3}}
\put(8.703,11){\circle*{0.3}}
\put(11.34,10){\circle*{0.3}}
\put(8.837,10){\circle*{0.3}}
\put(11.33,9){\circle*{0.3}}
\put(8.9652,9){\circle*{0.3}}
\put(11.32,8){\circle*{0.3}}
\put(9.0923,8){\circle*{0.3}}
\put(11.32,7){\circle*{0.3}}
\put(9.223,7){\circle*{0.3}}
\put(11.31,6){\circle*{0.3}}
\put(9.361,6){\circle*{0.3}}
\put(11.30,5){\circle*{0.3}}
\put(9.512,5){\circle*{0.3}}
\put(11.28,4){\circle*{0.3}}
\put(9.686,4){\circle*{0.3}}
\put(11.26,3){\circle*{0.3}}
\put(9.895,3){\circle*{0.3}}
\put(11.24,2){\circle*{0.3}}
\put(10.16,2){\circle*{0.3}}
\put(11.20,1){\circle*{0.3}}
\put(10.53,1){\circle*{0.3}}
\multiput(1,0.12)(1,0){8}{
\multiput(0,0)(0,0.2){105}{\circle*{0.07}}
}
\multiput(10,0.12)(1,0){8}{
\multiput(0,0)(0,0.2){105}{\circle*{0.07}}
}
\end{picture}
\end{center}
As we can confirm from the left one of the above diagrams, we have
\begin{align*}
\lambda_2^{(24)}>\lambda_2^{(25)}>\cdots >\lambda_2^{(44)},\\
\lambda_1^{(24)}<\lambda_1^{(25)}<\cdots <\lambda_1^{(44)},
\end{align*}
and from the right one of the above diagrams, we have
\begin{align*}
\lambda_2^{(45)}>\lambda_2^{(46)}>\cdots >\lambda_2^{(64)},\\
\lambda_1^{(45)}<\lambda_1^{(46)}<\cdots <\lambda_1^{(64)}.
\end{align*}
Moreover we can recognize clear gaps between three groups of the values:
\[
\{\lambda_2^{(1)},\lambda_2^{(2)},\ldots,\lambda_2^{(22)}\},\,
\{\lambda_2^{(24)},\lambda_2^{(25)},\ldots,\lambda_2^{(44)}\}\text{ and }
\{\lambda_2^{(45)},\lambda_2^{(46)},\ldots,\lambda_2^{(64)}\}.
\]
Therefore we conclude the following correspondences between the solutions
$\{\lambda_1^{(i)},\lambda_2^{(i)}\}$ and the rigged configurations $(r_1,r_2)$:
\begin{align*}
\{\lambda_1^{(24)},\lambda_2^{(24)}\}\longmapsto (20,0),\,\,
\{\lambda_1^{(25)},\lambda_2^{(25)}\}\longmapsto (20,1),\cdots,
\{\lambda_1^{(44)},\lambda_2^{(44)}\}\longmapsto (20,20),
\end{align*}
and
\begin{align*}
\{\lambda_1^{(45)},\lambda_2^{(45)}\}\longmapsto (19,0),\,\,
\{\lambda_1^{(46)},\lambda_2^{(46)}\}\longmapsto (19,1),\cdots,
\{\lambda_1^{(64)},\lambda_2^{(64)}\}\longmapsto (19,19).
\end{align*}

Let us return to the solution $\{\lambda_1^{(23)},\lambda_2^{(23)}\}$.
From the above analysis, we see that this solution is isolated and
does not fit into the rigged configurations picture.
Therefore the solution $\{\lambda_1^{(23)},\lambda_2^{(23)}\}$
is one of the exceptional real solutions.

Let us continue in a similar way.
The left one of the following diagrams shows the solutions
$\{\lambda_1^{(65)},\lambda_2^{(65)}\}$, $\cdots$, $\{\lambda_1^{(83)},\lambda_2^{(83)}\}$
and the right diagram shows the solutions
$\{\lambda_1^{(84)},\lambda_2^{(84)}\}$, $\cdots$, $\{\lambda_1^{(101)},\lambda_2^{(101)}\}$.
The first group corresponds to the rigged configurations $(18,0)$, $\cdots$, $(18,18)$
and the second group corresponds to $(17,0)$, $\cdots$, $(17,17)$, respectively.
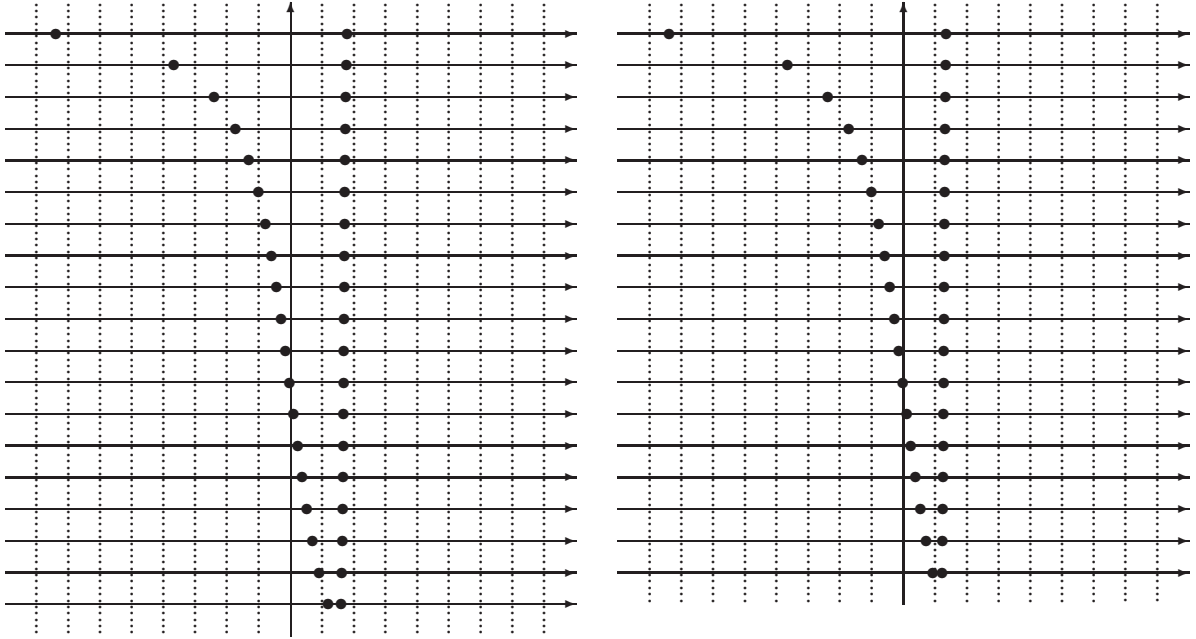
\begin{center}
\unitlength 12pt
\begin{picture}(19,20)
\put(9,0){\vector(0,1){20}}
\multiput(0,1)(0,1){19}{\vector(1,0){18}}
\put(10.78,19){\circle*{0.3}}
\put(1.60,19){\circle*{0.3}}
\put(10.76,18){\circle*{0.3}}
\put(5.32,18){\circle*{0.3}}
\put(10.74,17){\circle*{0.3}}
\put(6.60,17){\circle*{0.3}}
\put(10.73,16){\circle*{0.3}}
\put(7.27,16){\circle*{0.3}}
\put(10.72,15){\circle*{0.3}}
\put(7.69,15){\circle*{0.3}}
\put(10.71,14){\circle*{0.3}}
\put(7.99,14){\circle*{0.3}}
\put(10.71,13){\circle*{0.3}}
\put(8.217,13){\circle*{0.3}}
\put(10.70,12){\circle*{0.3}}
\put(8.405,12){\circle*{0.3}}
\put(10.70,11){\circle*{0.3}}
\put(8.566,11){\circle*{0.3}}
\put(10.69,10){\circle*{0.3}}
\put(8.710,10){\circle*{0.3}}
\put(10.68,9){\circle*{0.3}}
\put(8.843,9){\circle*{0.3}}
\put(10.68,8){\circle*{0.3}}
\put(8.9717,8){\circle*{0.3}}
\put(10.67,7){\circle*{0.3}}
\put(9.0993,7){\circle*{0.3}}
\put(10.67,6){\circle*{0.3}}
\put(9.230,6){\circle*{0.3}}
\put(10.66,5){\circle*{0.3}}
\put(9.370,5){\circle*{0.3}}
\put(10.65,4){\circle*{0.3}}
\put(9.523,4){\circle*{0.3}}
\put(10.64,3){\circle*{0.3}}
\put(9.700,3){\circle*{0.3}}
\put(10.62,2){\circle*{0.3}}
\put(9.914,2){\circle*{0.3}}
\put(10.60,1){\circle*{0.3}}
\put(10.19,1){\circle*{0.3}}
\multiput(1,0.12)(1,0){8}{
\multiput(0,0)(0,0.2){100}{\circle*{0.07}}
}
\multiput(10,0.12)(1,0){8}{
\multiput(0,0)(0,0.2){100}{\circle*{0.07}}
}
\end{picture}
\begin{picture}(18,20)(0,-1)
\put(9,0){\vector(0,1){19}}
\multiput(0,1)(0,1){18}{\vector(1,0){18}}
\put(10.35,18){\circle*{0.3}}
\put(1.62,18){\circle*{0.3}}
\put(10.33,17){\circle*{0.3}}
\put(5.34,17){\circle*{0.3}}
\put(10.32,16){\circle*{0.3}}
\put(6.61,16){\circle*{0.3}}
\put(10.31,15){\circle*{0.3}}
\put(7.28,15){\circle*{0.3}}
\put(10.30,14){\circle*{0.3}}
\put(7.70,14){\circle*{0.3}}
\put(10.30,13){\circle*{0.3}}
\put(7.99,13){\circle*{0.3}}
\put(10.29,12){\circle*{0.3}}
\put(8.223,12){\circle*{0.3}}
\put(10.29,11){\circle*{0.3}}
\put(8.410,11){\circle*{0.3}}
\put(10.28,10){\circle*{0.3}}
\put(8.571,10){\circle*{0.3}}
\put(10.28,9){\circle*{0.3}}
\put(8.715,9){\circle*{0.3}}
\put(10.27,8){\circle*{0.3}}
\put(8.848,8){\circle*{0.3}}
\put(10.27,7){\circle*{0.3}}
\put(8.9771,7){\circle*{0.3}}
\put(10.26,6){\circle*{0.3}}
\put(9.105,6){\circle*{0.3}}
\put(10.26,5){\circle*{0.3}}
\put(9.237,5){\circle*{0.3}}
\put(10.25,4){\circle*{0.3}}
\put(9.377,4){\circle*{0.3}}
\put(10.24,3){\circle*{0.3}}
\put(9.532,3){\circle*{0.3}}
\put(10.23,2){\circle*{0.3}}
\put(9.711,2){\circle*{0.3}}
\put(10.22,1){\circle*{0.3}}
\put(9.928,1){\circle*{0.3}}
\multiput(1,0.12)(1,0){8}{
\multiput(0,0)(0,0.2){95}{\circle*{0.07}}
}
\multiput(10,0.12)(1,0){8}{
\multiput(0,0)(0,0.2){95}{\circle*{0.07}}
}
\end{picture}
\end{center}

The left one of the following diagrams shows the solutions
$\{\lambda_1^{(102)},\lambda_2^{(102)}\}$, $\cdots$, $\{\lambda_1^{(118)},\lambda_2^{(118)}\}$
and the right diagram shows the solutions
$\{\lambda_1^{(119)},\lambda_2^{(119)}\}$, $\cdots$, $\{\lambda_1^{(134)},\lambda_2^{(134)}\}$.
The first group corresponds to the rigged configurations $(16,0)$, $\cdots$, $(16,16)$
and the second group corresponds to $(15,0)$, $\cdots$, $(15,15)$, respectively.
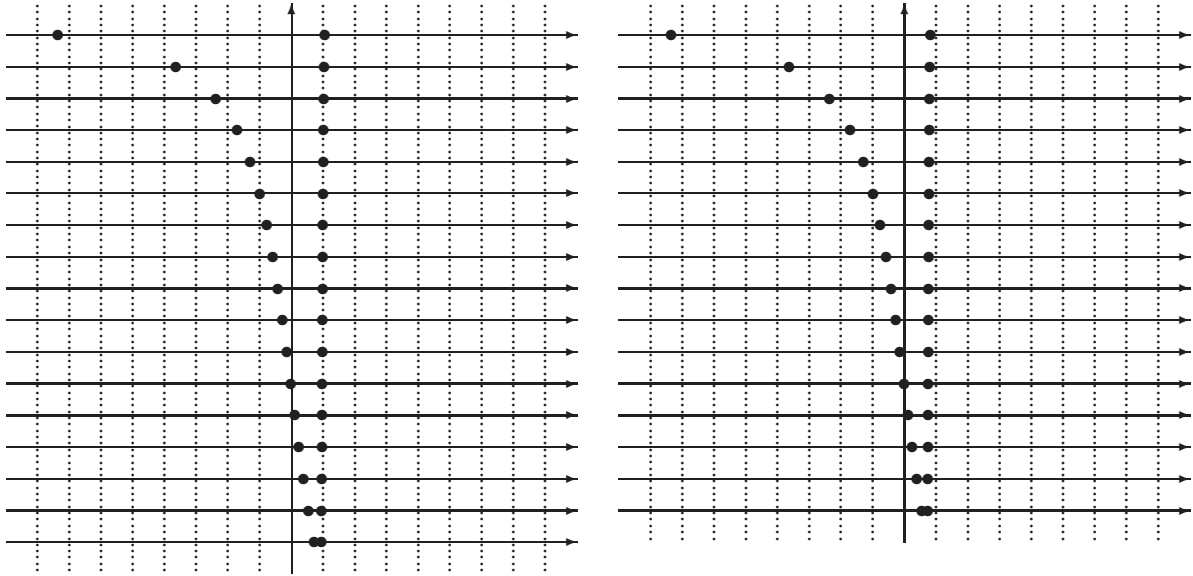
\begin{center}
\unitlength 12pt
\begin{picture}(19,18)
\put(9,0){\vector(0,1){18}}
\multiput(0,1)(0,1){17}{\vector(1,0){18}}
\put(10.05,17){\circle*{0.3}}
\put(1.64,17){\circle*{0.3}}
\put(10.03,16){\circle*{0.3}}
\put(5.35,16){\circle*{0.3}}
\put(10.02,15){\circle*{0.3}}
\put(6.62,15){\circle*{0.3}}
\put(10.01,14){\circle*{0.3}}
\put(7.29,14){\circle*{0.3}}
\put(10.01,13){\circle*{0.3}}
\put(7.70,13){\circle*{0.3}}
\put(10.00,12){\circle*{0.3}}
\put(8.00,12){\circle*{0.3}}
\put(9.995,11){\circle*{0.3}}
\put(8.228,11){\circle*{0.3}}
\put(9.991,10){\circle*{0.3}}
\put(8.415,10){\circle*{0.3}}
\put(9.986,9){\circle*{0.3}}
\put(8.575,9){\circle*{0.3}}
\put(9.982,8){\circle*{0.3}}
\put(8.719,8){\circle*{0.3}}
\put(9.978,7){\circle*{0.3}}
\put(8.853,7){\circle*{0.3}}
\put(9.974,6){\circle*{0.3}}
\put(8.9816,6){\circle*{0.3}}
\put(9.970,5){\circle*{0.3}}
\put(9.110,5){\circle*{0.3}}
\put(9.966,4){\circle*{0.3}}
\put(9.242,4){\circle*{0.3}}
\put(9.960,3){\circle*{0.3}}
\put(9.383,3){\circle*{0.3}}
\put(9.955,2){\circle*{0.3}}
\put(9.539,2){\circle*{0.3}}
\put(9.948,1){\circle*{0.3}}
\put(9.720,1){\circle*{0.3}}
\multiput(1,0.12)(1,0){8}{
\multiput(0,0)(0,0.2){90}{\circle*{0.07}}
}
\multiput(10,0.12)(1,0){8}{
\multiput(0,0)(0,0.2){90}{\circle*{0.07}}
}
\end{picture}
\begin{picture}(18,18)(0,-1)
\put(9,0){\vector(0,1){17}}
\multiput(0,1)(0,1){16}{\vector(1,0){18}}
\put(9.811,16){\circle*{0.3}}
\put(1.65,16){\circle*{0.3}}
\put(9.799,15){\circle*{0.3}}
\put(5.36,15){\circle*{0.3}}
\put(9.790,14){\circle*{0.3}}
\put(6.63,14){\circle*{0.3}}
\put(9.783,13){\circle*{0.3}}
\put(7.29,13){\circle*{0.3}}
\put(9.777,12){\circle*{0.3}}
\put(7.71,12){\circle*{0.3}}
\put(9.772,11){\circle*{0.3}}
\put(8.005,11){\circle*{0.3}}
\put(9.767,10){\circle*{0.3}}
\put(8.233,10){\circle*{0.3}}
\put(9.763,9){\circle*{0.3}}
\put(8.419,9){\circle*{0.3}}
\put(9.760,8){\circle*{0.3}}
\put(8.579,8){\circle*{0.3}}
\put(9.756,7){\circle*{0.3}}
\put(8.723,7){\circle*{0.3}}
\put(9.753,6){\circle*{0.3}}
\put(8.857,6){\circle*{0.3}}
\put(9.749,5){\circle*{0.3}}
\put(8.9854,5){\circle*{0.3}}
\put(9.745,4){\circle*{0.3}}
\put(9.114,4){\circle*{0.3}}
\put(9.741,3){\circle*{0.3}}
\put(9.246,3){\circle*{0.3}}
\put(9.737,2){\circle*{0.3}}
\put(9.388,2){\circle*{0.3}}
\put(9.732,1){\circle*{0.3}}
\put(9.545,1){\circle*{0.3}}
\multiput(1,0.12)(1,0){8}{
\multiput(0,0)(0,0.2){85}{\circle*{0.07}}
}
\multiput(10,0.12)(1,0){8}{
\multiput(0,0)(0,0.2){85}{\circle*{0.07}}
}
\end{picture}
\end{center}

The left one of the following diagrams shows the solutions
$\{\lambda_1^{(135)},\lambda_2^{(135)}\}$, $\cdots$, $\{\lambda_1^{(149)},\lambda_2^{(149)}\}$
and the right diagram shows the solutions
$\{\lambda_1^{(150)},\lambda_2^{(150)}\}$, $\cdots$, $\{\lambda_1^{(163)},\lambda_2^{(163)}\}$.
The first group corresponds to the rigged configurations $(14,0)$, $\cdots$, $(14,14)$
and the second group corresponds to $(13,0)$, $\cdots$, $(13,13)$, respectively.
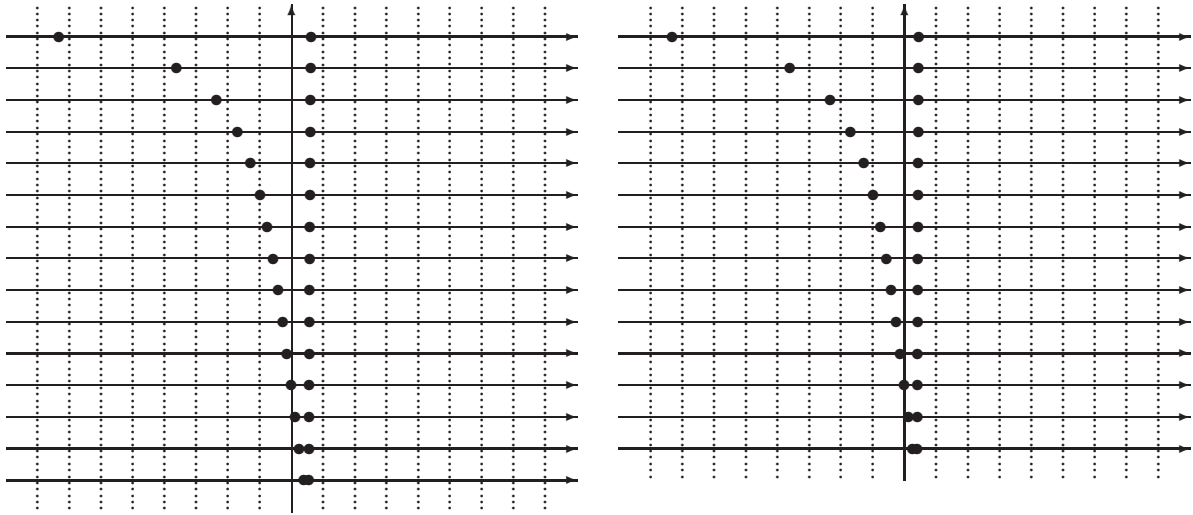
\begin{center}
\unitlength 12pt
\begin{picture}(19,16)
\put(9,0){\vector(0,1){16}}
\multiput(0,1)(0,1){15}{\vector(1,0){18}}
\put(9.621,15){\circle*{0.3}}
\put(1.67,15){\circle*{0.3}}
\put(9.610,14){\circle*{0.3}}
\put(5.37,14){\circle*{0.3}}
\put(9.602,13){\circle*{0.3}}
\put(6.64,13){\circle*{0.3}}
\put(9.595,12){\circle*{0.3}}
\put(7.30,12){\circle*{0.3}}
\put(9.590,11){\circle*{0.3}}
\put(7.71,11){\circle*{0.3}}
\put(9.585,10){\circle*{0.3}}
\put(8.009,10){\circle*{0.3}}
\put(9.581,9){\circle*{0.3}}
\put(8.237,9){\circle*{0.3}}
\put(9.577,8){\circle*{0.3}}
\put(8.423,8){\circle*{0.3}}
\put(9.574,7){\circle*{0.3}}
\put(8.583,7){\circle*{0.3}}
\put(9.571,6){\circle*{0.3}}
\put(8.726,6){\circle*{0.3}}
\put(9.568,5){\circle*{0.3}}
\put(8.860,5){\circle*{0.3}}
\put(9.564,4){\circle*{0.3}}
\put(8.9888,4){\circle*{0.3}}
\put(9.561,3){\circle*{0.3}}
\put(9.117,3){\circle*{0.3}}
\put(9.558,2){\circle*{0.3}}
\put(9.250,2){\circle*{0.3}}
\put(9.554,1){\circle*{0.3}}
\put(9.392,1){\circle*{0.3}}
\multiput(1,0.12)(1,0){8}{
\multiput(0,0)(0,0.2){80}{\circle*{0.07}}
}
\multiput(10,0.12)(1,0){8}{
\multiput(0,0)(0,0.2){80}{\circle*{0.07}}
}
\end{picture}
\begin{picture}(18,16)(0,-1)
\put(9,0){\vector(0,1){15}}
\multiput(0,1)(0,1){14}{\vector(1,0){18}}
\put(9.458,14){\circle*{0.3}}
\put(1.68,14){\circle*{0.3}}
\put(9.448,13){\circle*{0.3}}
\put(5.38,13){\circle*{0.3}}
\put(9.441,12){\circle*{0.3}}
\put(6.65,12){\circle*{0.3}}
\put(9.434,11){\circle*{0.3}}
\put(7.30,11){\circle*{0.3}}
\put(9.429,10){\circle*{0.3}}
\put(7.72,10){\circle*{0.3}}
\put(9.425,9){\circle*{0.3}}
\put(8.014,9){\circle*{0.3}}
\put(9.421,8){\circle*{0.3}}
\put(8.240,8){\circle*{0.3}}
\put(9.417,7){\circle*{0.3}}
\put(8.426,7){\circle*{0.3}}
\put(9.414,6){\circle*{0.3}}
\put(8.586,6){\circle*{0.3}}
\put(9.411,5){\circle*{0.3}}
\put(8.729,5){\circle*{0.3}}
\put(9.408,4){\circle*{0.3}}
\put(8.863,4){\circle*{0.3}}
\put(9.405,3){\circle*{0.3}}
\put(8.99184,3){\circle*{0.3}}
\put(9.402,2){\circle*{0.3}}
\put(9.121,2){\circle*{0.3}}
\put(9.399,1){\circle*{0.3}}
\put(9.254,1){\circle*{0.3}}
\multiput(1,0.12)(1,0){8}{
\multiput(0,0)(0,0.2){75}{\circle*{0.07}}
}
\multiput(10,0.12)(1,0){8}{
\multiput(0,0)(0,0.2){75}{\circle*{0.07}}
}
\end{picture}
\end{center}

The left one of the following diagrams shows the solutions
$\{\lambda_1^{(164)},\lambda_2^{(164)}\}$, $\cdots$, $\{\lambda_1^{(176)},\lambda_2^{(176)}\}$
and the right diagram shows the solutions
$\{\lambda_1^{(177)},\lambda_2^{(177)}\}$, $\cdots$, $\{\lambda_1^{(188)},\lambda_2^{(188)}\}$.
The first group corresponds to the rigged configurations $(12,0)$, $\cdots$, $(12,12)$
and the second group corresponds to $(11,0)$, $\cdots$, $(11,11)$, respectively.
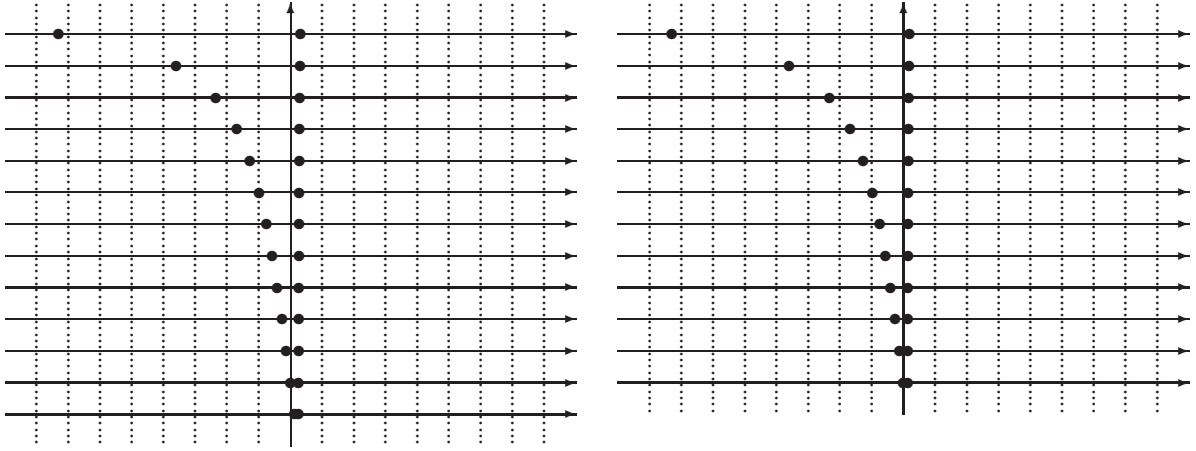
\begin{center}
\unitlength 12pt
\begin{picture}(19,14)
\put(9,0){\vector(0,1){14}}
\multiput(0,1)(0,1){13}{\vector(1,0){18}}
\put(9.314,13){\circle*{0.3}}
\put(1.69,13){\circle*{0.3}}
\put(9.304,12){\circle*{0.3}}
\put(5.39,12){\circle*{0.3}}
\put(9.297,11){\circle*{0.3}}
\put(6.65,11){\circle*{0.3}}
\put(9.290,10){\circle*{0.3}}
\put(7.31,10){\circle*{0.3}}
\put(9.285,9){\circle*{0.3}}
\put(7.72,9){\circle*{0.3}}
\put(9.281,8){\circle*{0.3}}
\put(8.018,8){\circle*{0.3}}
\put(9.277,7){\circle*{0.3}}
\put(8.244,7){\circle*{0.3}}
\put(9.274,6){\circle*{0.3}}
\put(8.429,6){\circle*{0.3}}
\put(9.271,5){\circle*{0.3}}
\put(8.589,5){\circle*{0.3}}
\put(9.268,4){\circle*{0.3}}
\put(8.732,4){\circle*{0.3}}
\put(9.265,3){\circle*{0.3}}
\put(8.866,3){\circle*{0.3}}
\put(9.262,2){\circle*{0.3}}
\put(8.99468,2){\circle*{0.3}}
\put(9.260,1){\circle*{0.3}}
\put(9.124,1){\circle*{0.3}}
\multiput(1,0.12)(1,0){8}{
\multiput(0,0)(0,0.2){70}{\circle*{0.07}}
}
\multiput(10,0.12)(1,0){8}{
\multiput(0,0)(0,0.2){70}{\circle*{0.07}}
}
\end{picture}
\begin{picture}(18,14)(0,-1)
\put(9,0){\vector(0,1){13}}
\multiput(0,1)(0,1){12}{\vector(1,0){18}}
\put(9.180,12){\circle*{0.3}}
\put(1.70,12){\circle*{0.3}}
\put(9.171,11){\circle*{0.3}}
\put(5.40,11){\circle*{0.3}}
\put(9.163,10){\circle*{0.3}}
\put(6.66,10){\circle*{0.3}}
\put(9.157,9){\circle*{0.3}}
\put(7.32,9){\circle*{0.3}}
\put(9.152,8){\circle*{0.3}}
\put(7.73,8){\circle*{0.3}}
\put(9.147,7){\circle*{0.3}}
\put(8.022,7){\circle*{0.3}}
\put(9.143,6){\circle*{0.3}}
\put(8.247,6){\circle*{0.3}}
\put(9.140,5){\circle*{0.3}}
\put(8.432,5){\circle*{0.3}}
\put(9.137,4){\circle*{0.3}}
\put(8.592,4){\circle*{0.3}}
\put(9.134,3){\circle*{0.3}}
\put(8.735,3){\circle*{0.3}}
\put(9.132,2){\circle*{0.3}}
\put(8.868,2){\circle*{0.3}}
\put(9.129,1){\circle*{0.3}}
\put(8.99737,1){\circle*{0.3}}
\multiput(1,0.12)(1,0){8}{
\multiput(0,0)(0,0.2){65}{\circle*{0.07}}
}
\multiput(10,0.12)(1,0){8}{
\multiput(0,0)(0,0.2){65}{\circle*{0.07}}
}
\end{picture}
\end{center}

The left one of the following diagrams shows the solutions
$\{\lambda_1^{(189)},\lambda_2^{(189)}\}$, $\cdots$, $\{\lambda_1^{(199)},\lambda_2^{(199)}\}$
and the right diagram shows the solutions
$\{\lambda_1^{(200)},\lambda_2^{(200)}\}$, $\cdots$, $\{\lambda_1^{(209)},\lambda_2^{(209)}\}$.
The first group corresponds to the rigged configurations $(10,0)$, $\cdots$, $(10,10)$
and the second group corresponds to $(9,0)$, $\cdots$, $(9,9)$, respectively.
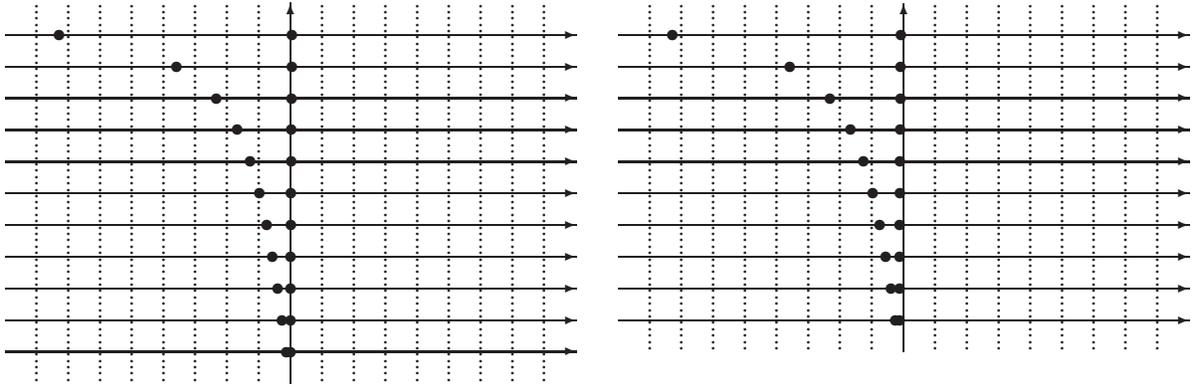
\begin{center}
\unitlength 12pt
\begin{picture}(19,12)
\put(9,0){\vector(0,1){12}}
\multiput(0,1)(0,1){11}{\vector(1,0){18}}
\put(9.0522,11){\circle*{0.3}}
\put(1.71,11){\circle*{0.3}}
\put(9.0427,10){\circle*{0.3}}
\put(5.41,10){\circle*{0.3}}
\put(9.0348,9){\circle*{0.3}}
\put(6.67,9){\circle*{0.3}}
\put(9.0283,8){\circle*{0.3}}
\put(7.32,8){\circle*{0.3}}
\put(9.0229,7){\circle*{0.3}}
\put(7.73,7){\circle*{0.3}}
\put(9.0184,6){\circle*{0.3}}
\put(8.026,6){\circle*{0.3}}
\put(9.0146,5){\circle*{0.3}}
\put(8.251,5){\circle*{0.3}}
\put(9.0112,4){\circle*{0.3}}
\put(8.436,4){\circle*{0.3}}
\put(9.00816,3){\circle*{0.3}}
\put(8.595,3){\circle*{0.3}}
\put(9.00532,2){\circle*{0.3}}
\put(8.738,2){\circle*{0.3}}
\put(9.00263,1){\circle*{0.3}}
\put(8.871,1){\circle*{0.3}}
\multiput(1,0.12)(1,0){8}{
\multiput(0,0)(0,0.2){60}{\circle*{0.07}}
}
\multiput(10,0.12)(1,0){8}{
\multiput(0,0)(0,0.2){60}{\circle*{0.07}}
}
\end{picture}
\begin{picture}(18,12)(0,-1)
\put(9,0){\vector(0,1){11}}
\multiput(0,1)(0,1){10}{\vector(1,0){18}}
\put(8.9262,10){\circle*{0.3}}
\put(1.72,10){\circle*{0.3}}
\put(8.9162,9){\circle*{0.3}}
\put(5.42,9){\circle*{0.3}}
\put(8.9077,8){\circle*{0.3}}
\put(6.68,8){\circle*{0.3}}
\put(8.9007,7){\circle*{0.3}}
\put(7.33,7){\circle*{0.3}}
\put(8.895,6){\circle*{0.3}}
\put(7.74,6){\circle*{0.3}}
\put(8.890,5){\circle*{0.3}}
\put(8.030,5){\circle*{0.3}}
\put(8.886,4){\circle*{0.3}}
\put(8.255,4){\circle*{0.3}}
\put(8.883,3){\circle*{0.3}}
\put(8.439,3){\circle*{0.3}}
\put(8.879,2){\circle*{0.3}}
\put(8.598,2){\circle*{0.3}}
\put(8.876,1){\circle*{0.3}}
\put(8.740,1){\circle*{0.3}}
\multiput(1,0.12)(1,0){8}{
\multiput(0,0)(0,0.2){55}{\circle*{0.07}}
}
\multiput(10,0.12)(1,0){8}{
\multiput(0,0)(0,0.2){55}{\circle*{0.07}}
}
\end{picture}
\end{center}

The left one of the following diagrams shows the solutions
$\{\lambda_1^{(210)},\lambda_2^{(210)}\}$, $\cdots$, $\{\lambda_1^{(218)},\lambda_2^{(218)}\}$
and the right diagram shows the solutions
$\{\lambda_1^{(219)},\lambda_2^{(219)}\}$, $\cdots$, $\{\lambda_1^{(226)},\lambda_2^{(226)}\}$.
The first group corresponds to the rigged configurations $(8,0)$, $\cdots$, $(8,8)$
and the second group corresponds to $(7,0)$, $\cdots$, $(7,7)$, respectively.
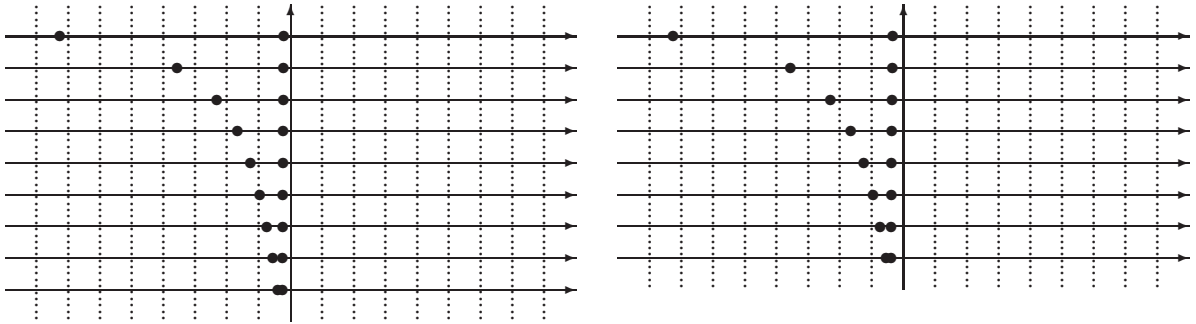
\begin{center}
\unitlength 12pt
\begin{picture}(19,10)
\put(9,0){\vector(0,1){10}}
\multiput(0,1)(0,1){9}{\vector(1,0){18}}
\put(8.798,9){\circle*{0.3}}
\put(1.73,9){\circle*{0.3}}
\put(8.787,8){\circle*{0.3}}
\put(5.43,8){\circle*{0.3}}
\put(8.777,7){\circle*{0.3}}
\put(6.68,7){\circle*{0.3}}
\put(8.770,6){\circle*{0.3}}
\put(7.33,6){\circle*{0.3}}
\put(8.763,5){\circle*{0.3}}
\put(7.74,5){\circle*{0.3}}
\put(8.758,4){\circle*{0.3}}
\put(8.034,4){\circle*{0.3}}
\put(8.754,3){\circle*{0.3}}
\put(8.259,3){\circle*{0.3}}
\put(8.750,2){\circle*{0.3}}
\put(8.442,2){\circle*{0.3}}
\put(8.746,1){\circle*{0.3}}
\put(8.601,1){\circle*{0.3}}
\multiput(1,0.12)(1,0){8}{
\multiput(0,0)(0,0.2){50}{\circle*{0.07}}
}
\multiput(10,0.12)(1,0){8}{
\multiput(0,0)(0,0.2){50}{\circle*{0.07}}
}
\end{picture}
\begin{picture}(18,10)(0,-1)
\put(9,0){\vector(0,1){9}}
\multiput(0,1)(0,1){8}{\vector(1,0){18}}
\put(8.663,8){\circle*{0.3}}
\put(1.74,8){\circle*{0.3}}
\put(8.650,7){\circle*{0.3}}
\put(5.44,7){\circle*{0.3}}
\put(8.639,6){\circle*{0.3}}
\put(6.69,6){\circle*{0.3}}
\put(8.630,5){\circle*{0.3}}
\put(7.34,5){\circle*{0.3}}
\put(8.623,4){\circle*{0.3}}
\put(7.75,4){\circle*{0.3}}
\put(8.617,3){\circle*{0.3}}
\put(8.040,3){\circle*{0.3}}
\put(8.612,2){\circle*{0.3}}
\put(8.263,2){\circle*{0.3}}
\put(8.608,1){\circle*{0.3}}
\put(8.446,1){\circle*{0.3}}
\multiput(1,0.12)(1,0){8}{
\multiput(0,0)(0,0.2){45}{\circle*{0.07}}
}
\multiput(10,0.12)(1,0){8}{
\multiput(0,0)(0,0.2){45}{\circle*{0.07}}
}
\end{picture}
\end{center}

The left one of the following diagrams shows the solutions
$\{\lambda_1^{(227)},\lambda_2^{(227)}\}$, $\cdots$, $\{\lambda_1^{(233)},\lambda_2^{(233)}\}$
and the right diagram shows the solutions
$\{\lambda_1^{(234)},\lambda_2^{(234)}\}$, $\cdots$, $\{\lambda_1^{(239)},\lambda_2^{(239)}\}$.
The first group corresponds to the rigged configurations $(6,0)$, $\cdots$, $(6,6)$
and the second group corresponds to $(5,0)$, $\cdots$, $(5,5)$, respectively.
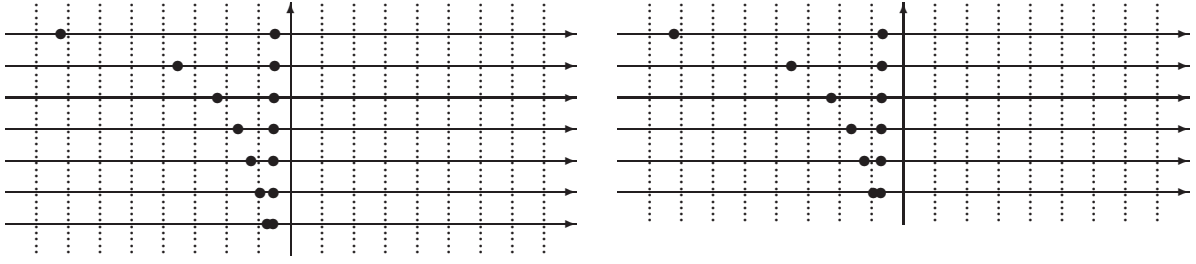
\begin{center}
\unitlength 12pt
\begin{picture}(19,8)
\put(9,0){\vector(0,1){8}}
\multiput(0,1)(0,1){7}{\vector(1,0){18}}
\put(8.515,7){\circle*{0.3}}
\put(1.76,7){\circle*{0.3}}
\put(8.501,6){\circle*{0.3}}
\put(5.45,6){\circle*{0.3}}
\put(8.488,5){\circle*{0.3}}
\put(6.70,5){\circle*{0.3}}
\put(8.477,4){\circle*{0.3}}
\put(7.35,4){\circle*{0.3}}
\put(8.468,3){\circle*{0.3}}
\put(7.76,3){\circle*{0.3}}
\put(8.461,2){\circle*{0.3}}
\put(8.045,2){\circle*{0.3}}
\put(8.455,1){\circle*{0.3}}
\put(8.268,1){\circle*{0.3}}
\multiput(1,0.12)(1,0){8}{
\multiput(0,0)(0,0.2){40}{\circle*{0.07}}
}
\multiput(10,0.12)(1,0){8}{
\multiput(0,0)(0,0.2){40}{\circle*{0.07}}
}
\end{picture}
\begin{picture}(18,8)(0,-1)
\put(9,0){\vector(0,1){7}}
\multiput(0,1)(0,1){6}{\vector(1,0){18}}
\put(8.349,6){\circle*{0.3}}
\put(1.77,6){\circle*{0.3}}
\put(8.331,5){\circle*{0.3}}
\put(5.47,5){\circle*{0.3}}
\put(8.314,4){\circle*{0.3}}
\put(6.72,4){\circle*{0.3}}
\put(8.300,3){\circle*{0.3}}
\put(7.36,3){\circle*{0.3}}
\put(8.289,2){\circle*{0.3}}
\put(7.77,2){\circle*{0.3}}
\put(8.280,1){\circle*{0.3}}
\put(8.052,1){\circle*{0.3}}
\multiput(1,0.12)(1,0){8}{
\multiput(0,0)(0,0.2){35}{\circle*{0.07}}
}
\multiput(10,0.12)(1,0){8}{
\multiput(0,0)(0,0.2){35}{\circle*{0.07}}
}
\end{picture}
\end{center}

The left one of the following diagrams shows the solutions
$\{\lambda_1^{(240)},\lambda_2^{(240)}\}$, $\cdots$, $\{\lambda_1^{(244)},\lambda_2^{(244)}\}$
and the right diagram shows the solutions
$\{\lambda_1^{(245)},\lambda_2^{(245)}\}$, $\cdots$, $\{\lambda_1^{(248)},\lambda_2^{(248)}\}$.
The first group corresponds to the rigged configurations $(4,0)$, $\cdots$, $(4,4)$
and the second group corresponds to $(3,0)$, $\cdots$, $(3,3)$, respectively.
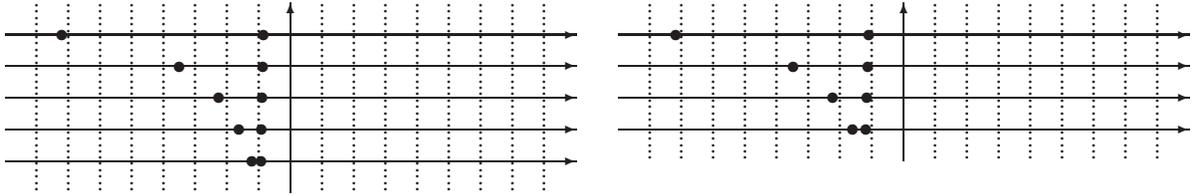
\begin{center}
\unitlength 12pt
\begin{picture}(19,6)
\put(9,0){\vector(0,1){6}}
\multiput(0,1)(0,1){5}{\vector(1,0){18}}
\put(8.152,5){\circle*{0.3}}
\put(1.79,5){\circle*{0.3}}
\put(8.128,4){\circle*{0.3}}
\put(5.49,4){\circle*{0.3}}
\put(8.105,3){\circle*{0.3}}
\put(6.74,3){\circle*{0.3}}
\put(8.086,2){\circle*{0.3}}
\put(7.38,2){\circle*{0.3}}
\put(8.072,1){\circle*{0.3}}
\put(7.78,1){\circle*{0.3}}
\multiput(1,0.12)(1,0){8}{
\multiput(0,0)(0,0.2){30}{\circle*{0.07}}
}
\multiput(10,0.12)(1,0){8}{
\multiput(0,0)(0,0.2){30}{\circle*{0.07}}
}
\end{picture}
\begin{picture}(18,6)(0,-1)
\put(9,0){\vector(0,1){5}}
\multiput(0,1)(0,1){4}{\vector(1,0){18}}
\put(7.91,4){\circle*{0.3}}
\put(1.82,4){\circle*{0.3}}
\put(7.87,3){\circle*{0.3}}
\put(5.52,3){\circle*{0.3}}
\put(7.84,2){\circle*{0.3}}
\put(6.76,2){\circle*{0.3}}
\put(7.81,1){\circle*{0.3}}
\put(7.40,1){\circle*{0.3}}
\multiput(1,0.12)(1,0){8}{
\multiput(0,0)(0,0.2){25}{\circle*{0.07}}
}
\multiput(10,0.12)(1,0){8}{
\multiput(0,0)(0,0.2){25}{\circle*{0.07}}
}
\end{picture}
\end{center}

The left one of the following diagrams shows the solutions
$\{\lambda_1^{(249)},\lambda_2^{(249)}\}$, $\{\lambda_1^{(250)},\lambda_2^{(250)}\}$
and $\{\lambda_1^{(251)},\lambda_2^{(251)}\}$
and the right diagram shows the solutions
$\{\lambda_1^{(252)},\lambda_2^{(252)}\}$ and $\{\lambda_1^{(253)},\lambda_2^{(253)}\}$.
The first group corresponds to the rigged configurations $(2,0)$, $(2,1)$ and $(2,2)$
and the second group corresponds to $(1,0)$ and $(1,1)$, respectively.
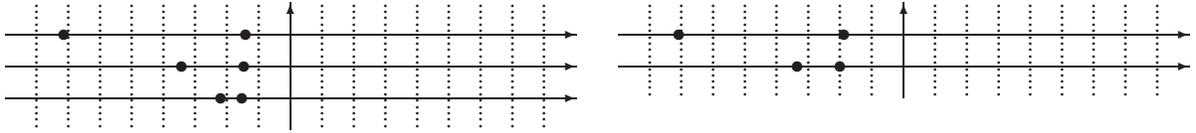
\begin{center}
\unitlength 12pt
\begin{picture}(19,4)
\put(9,0){\vector(0,1){4}}
\multiput(0,1)(0,1){3}{\vector(1,0){18}}
\put(7.58,3){\circle*{0.3}}
\put(1.86,3){\circle*{0.3}}
\put(7.53,2){\circle*{0.3}}
\put(5.56,2){\circle*{0.3}}
\put(7.47,1){\circle*{0.3}}
\put(6.80,1){\circle*{0.3}}
\multiput(1,0.12)(1,0){8}{
\multiput(0,0)(0,0.2){20}{\circle*{0.07}}
}
\multiput(10,0.12)(1,0){8}{
\multiput(0,0)(0,0.2){20}{\circle*{0.07}}
}
\end{picture}
\begin{picture}(18,4)(0,-1)
\put(9,0){\vector(0,1){3}}
\multiput(0,1)(0,1){2}{\vector(1,0){18}}
\put(7.12,2){\circle*{0.3}}
\put(1.92,2){\circle*{0.3}}
\put(7.00,1){\circle*{0.3}}
\put(5.65,1){\circle*{0.3}}
\multiput(1,0.12)(1,0){8}{
\multiput(0,0)(0,0.2){15}{\circle*{0.07}}
}
\multiput(10,0.12)(1,0){8}{
\multiput(0,0)(0,0.2){15}{\circle*{0.07}}
}
\end{picture}
\end{center}

The left one of the following diagrams shows the solution $\{\lambda_1^{(254)},\lambda_2^{(254)}\}$
and the right diagram shows the solution $\{\lambda_1^{(255)},\lambda_2^{(255)}\}$.
The first solution corresponds to the rigged configuration $(0,0)$ and
the second solution corresponds to the exceptional real solution.
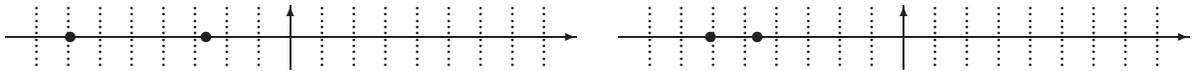
\begin{center}
\unitlength 12pt
\begin{picture}(19,2)
\put(9,0){\vector(0,1){2}}
\multiput(0,1)(0,1){1}{\vector(1,0){18}}
\put(6.34,1){\circle*{0.3}}
\put(2.06,1){\circle*{0.3}}
\multiput(1,0.12)(1,0){8}{
\multiput(0,0)(0,0.2){10}{\circle*{0.07}}
}
\multiput(10,0.12)(1,0){8}{
\multiput(0,0)(0,0.2){10}{\circle*{0.07}}
}
\end{picture}
\begin{picture}(18,2)
\put(9,0){\vector(0,1){2}}
\multiput(0,1)(0,1){1}{\vector(1,0){18}}
\put(4.40,1){\circle*{0.3}}
\put(2.92,1){\circle*{0.3}}
\multiput(1,0.12)(1,0){8}{
\multiput(0,0)(0,0.2){10}{\circle*{0.07}}
}
\multiput(10,0.12)(1,0){8}{
\multiput(0,0)(0,0.2){10}{\circle*{0.07}}
}
\end{picture}
\end{center}

Some readers might feel that the identification of the solution $\{\lambda_1^{(255)},\lambda_2^{(255)}\}$
with the exceptional real solution is somewhat tricky.
In this case, we introduce another order based on the values of the smaller roots $\lambda_1^{(i)}$.
Then the solution $\{\lambda_1^{(255)},\lambda_2^{(255)}\}$ naturally appears
as the exceptional real solution.

Next, let us consider the complex solutions.
Among 21 such solutions, there is the solution $\{i/2,-i/2\}$.
Since this solution does not satisfy the condition (\ref{eq:NepomechieWangCriterion}), this is unphysical.
Therefore we have to consider the remaining 20 complex solutions
together with the 2 real solutions detected in the previous step.
We show these 22 solutions on the complex plain as follows
(the spacing of the dotted lines is 0.5).
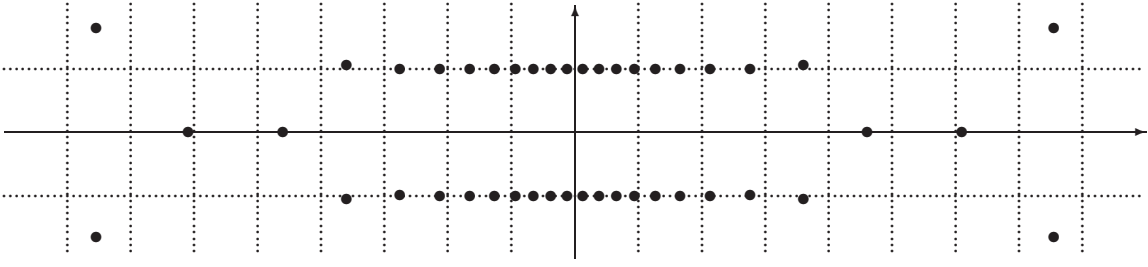
\begin{center}
\unitlength 12pt
\begin{picture}(36,8)
\put(18,0){\vector(0,1){8}}
\put(0,4){\vector(1,0){36}}
\put(30.2,4){\circle*{0.3}}
\put(27.21,4){\circle*{0.3}}
\put(8.79,4){\circle*{0.3}}
\put(5.8,4){\circle*{0.3}}
\put(2.9,0.7){\circle*{0.3}}
\put(2.9,7.3){\circle*{0.3}}
\put(10.79,1.88){\circle*{0.3}}
\put(10.79,6.12){\circle*{0.3}}
\put(12.48,2.02){\circle*{0.3}}
\put(12.48,5.98){\circle*{0.3}}
\put(13.74,2.00){\circle*{0.3}}
\put(13.74,6.00){\circle*{0.3}}
\put(14.69,2.00){\circle*{0.3}}
\put(14.69,6.00){\circle*{0.3}}
\put(15.46,2.00){\circle*{0.3}}
\put(15.46,6.00){\circle*{0.3}}
\put(16.12,2.00){\circle*{0.3}}
\put(16.12,6.00){\circle*{0.3}}
\put(16.70,2.00){\circle*{0.3}}
\put(16.70,6.00){\circle*{0.3}}
\put(17.24,2.00){\circle*{0.3}}
\put(17.24,6.00){\circle*{0.3}}
\put(17.75,2.00){\circle*{0.3}}
\put(17.75,6.00){\circle*{0.3}}
\put(18.25,2.00){\circle*{0.3}}
\put(18.25,6.00){\circle*{0.3}}
\put(18.76,2.00){\circle*{0.3}}
\put(18.76,6.00){\circle*{0.3}}
\put(19.30,2.00){\circle*{0.3}}
\put(19.30,6.00){\circle*{0.3}}
\put(19.88,2.00){\circle*{0.3}}
\put(19.88,6.00){\circle*{0.3}}
\put(20.54,2.00){\circle*{0.3}}
\put(20.54,6.00){\circle*{0.3}}
\put(21.31,2.00){\circle*{0.3}}
\put(21.31,6.00){\circle*{0.3}}
\put(22.26,2.00){\circle*{0.3}}
\put(22.26,6.00){\circle*{0.3}}
\put(23.52,2.02){\circle*{0.3}}
\put(23.52,5.98){\circle*{0.3}}
\put(25.21,1.88){\circle*{0.3}}
\put(25.21,6.12){\circle*{0.3}}
\put(33.1,0.7){\circle*{0.3}}
\put(33.1,7.3){\circle*{0.3}}
\multiput(2,0.24)(2,0){8}{
\multiput(0,0)(0,0.2){40}{\circle*{0.07}}
}
\multiput(20,0.24)(2,0){8}{
\multiput(0,0)(0,0.2){40}{\circle*{0.07}}
}
\multiput(0,2)(0.2,0){180}{\circle*{0.07}}
\multiput(0,6)(0.2,0){180}{\circle*{0.07}}
\end{picture}
\end{center}
Here the complex solutions take the form of 2-strings:
\[
\lambda_1=a+\left(\frac{1}{2}+\delta\right)i,\qquad
\lambda_2=a-\left(\frac{1}{2}+\delta\right)i
\]
where $a,\delta\in\mathbb{R}$.
These solutions should correspond to the rigged configurations
\begin{center}
\unitlength 12pt
\begin{picture}(4.5,1)
\put(0,0.1){21}
\put(1.2,0){\Yboxdim12pt\yng(2)}
\put(3.5,0.1){$r$}
\end{picture}
\end{center}
where $0\leq r\leq 21$.
Again, we can construct a correspondence based on the real parts of $\{\lambda_1,\lambda_2\}$
without any ambiguity.
In particular, the exceptional real solutions corresponds to the following rigged configurations:
\begin{center}
\unitlength 12pt
\begin{picture}(10,1)
\put(0,0.1){21}
\put(1.2,0){\Yboxdim12pt\yng(2)}
\put(3.5,0.1){1}
\put(6,0){
\put(0,0.1){21}
\put(1.2,0){\Yboxdim12pt\yng(2)}
\put(3.5,0.1){20}
}
\end{picture}
\end{center}

Let us compare the above analysis with Appendix of \cite{EKS:1992}.
For the comparison we have to divide their $\Lambda^1_1$ and $\Lambda^1_2$ by 2.
In their Table 2, they consider real solutions.
In our analysis, the two row rigged configurations $(r_1,r_2)$ corresponding to their solutions
are $(11,10)$, $(12,9)$, $(13,8)$, $(14,7)$, $(15,6)$, $(16,5)$, $(17,4)$, $(18,3)$,
$(19,2)$, $(20,1)$ and $(21,0)$ from top to bottom of their table.\footnote{We expect there is a typo at
line 2 from bottom of their table.}
In their Table 3, they consider the complex solutions.
According to our analysis, their solutions correspond to the one row rigged configurations
of the riggings $21,19,18,17\ldots,4,3,2,0$ from top to bottom of their table.
In Table 4, they consider real solutions.
According to our analysis, their solutions correspond to the two row rigged configurations
with the riggings $(21,21)$, $(21,20)$, $(21,19)$, $(21,18)$, $(21,17)$, $(21,16)$,
$(21,15)$, $(21,14)$, $(21,13)$, $(21,12)$, $(21,11)$, $(21,0)$, $(21,1)$, $(21,2)$, $(21,3)$,
$(21,4)$, $(21,5)$, $(21,6)$, $(21,7)$, $(21,8)$, $(21,9)$, $(21,10)$ from top to bottom of their table.
Therefore, under a suitable transformation, our analysis agrees with their Tables 2, 3 and 4
based on the Bethe--Takahashi quantum number.
However our interpretation of the exceptional real solutions are different from
their interpretation given in Table 1 of their paper.
A possible reasoning for this discrepancy might be the following.
In the derivation of the Bethe--Takahashi quantum number \cite{Takahashi},
we assume that the solutions take the form of a $(2b+1)$-string:
\[
a+bi,\,a+(b-1)i,\,a+(b-2)i,\,\ldots,\,a-bi,\qquad
(a\in\mathbb{R},\,b\in\mathbb{Z}_{\geq 0}/2).
\]
However the exceptional real solutions take very different form from the standard 2-strings.

According to \cite{EKS:1992} we have additional exceptional real solutions when $N>61.34$.
We expect we can analyze such cases in a similar way.
However currently I am unable to construct complete solutions for such cases.

\appendix

\section{Proofs for the assertions in Section \ref{sec:singular_solution}}
\label{sec:appendix}
Since the proofs for the assertions in Section \ref{sec:singular_solution}
are scattered around several places, we believe it is worthwhile to provide
proofs for the relevant assertions in a systematic way.

Recall that from the relation of the transfer matrix $T_N(\lambda)$ with the $R$-matrices
(known as the $RTT=TTR$ relation),
one can derive the commutation relations between the fundamental operators
$A_N$, $B_N$ and $D_N$ (see \cite{Faddeev,KorepinBook}):
\begin{align}
\label{eq:Bethe_relations1}
&[B_N(\lambda),B_N(\mu)]=0,\\
&\label{eq:Bethe_relations2}
A_N(\lambda)B_N(\mu)=\frac{\lambda-\mu-i}{\lambda-\mu}B_N(\mu)A_N(\lambda)
+\frac{i}{\lambda-\mu}B_N(\lambda)A_N(\mu),\\
&\label{eq:Bethe_relations3}
D_N(\lambda)B_N(\mu)=\frac{\lambda-\mu+i}{\lambda-\mu}B_N(\mu)D_N(\lambda)
-\frac{i}{\lambda-\mu}B_N(\lambda)D_N(\mu).
\end{align}

Let
\begin{align}
\Psi^{(\epsilon)}_\lambda:=\frac{1}{\epsilon^N}
B_N\!\left(\frac{i}{2}+\epsilon+c\,\epsilon^N\right)B_N\!\left(-\frac{i}{2}+\epsilon\right)
B_N(\lambda_3)\cdots B_N(\lambda_\ell)
|0\rangle_N.
\end{align}

\begin{lemma}
The vector $\lim_{\epsilon\rightarrow 0}\Psi^{(\epsilon)}_\lambda$ is well-defined.
\end{lemma}
\begin{proof}
Since $B_N$ operators commute with each other, it is enough to prove
\begin{align}
\label{eq:NepomechieWangAppendix_main}
B_N\!\left(\frac{i}{2}+\epsilon+c\,\epsilon^N\right)B_N\!\left(-\frac{i}{2}+\epsilon\right)
|0\rangle_N
\sim\epsilon^N
\end{align}
where $\sim$ means the order of the leading term in the expansion of $\epsilon$.
Our argument is similar to that given in \cite[Appendix A]{NepomechieWang2013}.

The proof is by induction on $N$.
As the starting point, $N=4$ case is checked explicitly as in (\ref{eq:NepomechieWang_N=4}).
Suppose that we have checked (\ref{eq:NepomechieWangAppendix_main}) for $N-1$ case.
Let us denote
\[
|0\rangle_N=|0\rangle_{N-1}\otimes
\left(\!
\begin{array}{c}
1\\ 0
\end{array}
\!\right)_{N}.
\]
From the definition of the transfer matrix at (\ref{def:transfer_matrix}), we see
that $T_N$ and $T_{N-1}$ are related as
\[
T_N(\lambda)=L_N(\lambda)T_{N-1}(\lambda).
\]
By definition of $L_N(\lambda)$ at (\ref{eq:def_L_of_Bethe}), we have
\[
\left(
\begin{array}{cc}
A_N(\lambda) & B_N(\lambda)\\
C_N(\lambda) & D_N(\lambda)
\end{array}
\right)
=
\left(
\begin{array}{cc}
a_N(\lambda) & b_N(\lambda)\\
c_N(\lambda) & d_N(\lambda)
\end{array}
\right)
\left(
\begin{array}{cc}
A_{N-1}(\lambda) & B_{N-1}(\lambda)\\
C_{N-1}(\lambda) & D_{N-1}(\lambda)
\end{array}
\right)
\]
where the operators $a_N(\lambda)$, $b_N(\lambda)$, $c_N(\lambda)$ and $d_N(\lambda)$
act on
$\left(\!
\begin{array}{c}
1\\ 0
\end{array}
\!\right)_{N}$
as
\begin{align*}
a_N(\lambda)&=
\left(
\begin{array}{cc}
\lambda+\frac{i}{2} & 0\\
0 & \lambda-\frac{i}{2}
\end{array}
\right),
&
b_N(\lambda)&=
\left(
\begin{array}{cc}
0 & 0\\
i & 0
\end{array}
\right),\\
c_N(\lambda)&=
\left(
\begin{array}{cc}
0 & i\\
0 & 0
\end{array}
\right),
&
d_N(\lambda)&=
\left(
\begin{array}{cc}
\lambda-\frac{i}{2} & 0\\
0 & \lambda+\frac{i}{2}
\end{array}
\right)
\end{align*}
and the operators $A_{N-1}(\lambda)$, $B_{N-1}(\lambda)$, $C_{N-1}(\lambda)$
and $D_{N-1}(\lambda)$ act on $|0\rangle_{N-1}$.
In particular,
\[
B_N(\lambda)=
a_N(\lambda)B_{N-1}(\lambda)+b_N(\lambda)D_{N-1}(\lambda).
\]
It follows that
\begin{align}
\nonumber
&B_N(\lambda_1)B_N(\lambda_2)|0\rangle_N\\
\nonumber
&=\left\{
a_N(\lambda_1)B_{N-1}(\lambda_1)+b_N(\lambda_1)D_{N-1}(\lambda_1)
\right\}\\
\nonumber
&\quad\times
\left\{
a_N(\lambda_2)B_{N-1}(\lambda_2)+b_N(\lambda_2)D_{N-1}(\lambda_2)
\right\}|0\rangle_N\\
\nonumber
&=\{ a_N(\lambda_1)a_N(\lambda_2)B_{N-1}(\lambda_1)B_{N-1}(\lambda_2)\\
\nonumber
&\quad
+a_N(\lambda_1)b_N(\lambda_2)B_{N-1}(\lambda_1)D_{N-1}(\lambda_2)\\
\nonumber
&\quad
+b_N(\lambda_1)a_N(\lambda_2)D_{N-1}(\lambda_1)B_{N-1}(\lambda_2)\\
\label{eq:NepomechieWangAppendix}
&\quad
+b_N(\lambda_1)b_N(\lambda_2)D_{N-1}(\lambda_1)D_{N-1}(\lambda_2)
\}|0\rangle_N
\end{align}
for $\lambda_1$, $\lambda_2$ arbitrary.

We now set $\lambda_1=\frac{i}{2}+\epsilon+c\,\epsilon^N$
and $\lambda_2=-\frac{i}{2}+\epsilon$, and we consider the four terms
of the right hand side of (\ref{eq:NepomechieWangAppendix}) respectively.
In the first term, we have
\[
B_{N-1}(\lambda_1)B_{N-1}(\lambda_2)|0\rangle_{N-1}\sim\epsilon^{N-1}
\]
by the induction hypothesis, and we have
\[
a_N(\lambda_1)a_N(\lambda_2)
\left(\!
\begin{array}{c}
1\\ 0
\end{array}
\!\right)_{N}
=
\left(\!
\begin{array}{cc}
\epsilon(i+\epsilon+c\,\epsilon^N)&0\\
0&(\epsilon+c\,\epsilon^N)(-i+\epsilon)
\end{array}
\!\right)
\left(\!
\begin{array}{c}
1\\ 0
\end{array}
\!\right)_{N}
\sim\epsilon .
\]
Hence the first term on the right hand side of (\ref{eq:NepomechieWangAppendix})
is of order $\epsilon^N$.

The fourth term on the right hand side of (\ref{eq:NepomechieWangAppendix})
is zero since we have
\[
b_N(\lambda_1)b_N(\lambda_2)
\left(\!
\begin{array}{c}
1\\ 0
\end{array}
\!\right)_{N}=0.
\]

From (\ref{eq:Bethe_relations3}), we have
\begin{align*}
D_{N-1}(\lambda_1)B_{N-1}(\lambda_2)&
=\frac{\lambda_1-\lambda_2+i}{\lambda_1-\lambda_2}B_{N-1}(\lambda_2)D_{N-1}(\lambda_1)
-\frac{i}{\lambda_1-\lambda_2}B_{N-1}(\lambda_1)D_{N-1}(\lambda_2)\\
&=
\frac{2i+c\,\epsilon^N}{i+c\,\epsilon^N}B_{N-1}(\lambda_2)D_{N-1}(\lambda_1)
-\frac{i}{i+c\,\epsilon^N}B_{N-1}(\lambda_1)D_{N-1}(\lambda_2).
\end{align*}
We apply this relation to the third term on the right hand side of (\ref{eq:NepomechieWangAppendix}).
Then the combination of the second term and the third term is
\begin{align}
\nonumber
&\left\{
a_N(\lambda_1)b_N(\lambda_2)-\frac{i}{i+c\,\epsilon^N}b_N(\lambda_1)a_N(\lambda_2)
\right\}
B_{N-1}(\lambda_1)D_{N-1}(\lambda_2)|0\rangle_N\\
\label{eq:NepomechieWangAppendix2}
&+\frac{2i+c\,\epsilon^N}{i+c\,\epsilon^N}b_N(\lambda_1)a_N(\lambda_2)
B_{N-1}(\lambda_2)D_{N-1}(\lambda_1)|0\rangle_N.
\end{align}
The first term of (\ref{eq:NepomechieWangAppendix2}) is
\begin{align*}
&\left\{
a_N(\lambda_1)b_N(\lambda_2)-\frac{i}{i+c\,\epsilon^N}b_N(\lambda_1)a_N(\lambda_2)
\right\}
\left(\!
\begin{array}{c}
1\\ 0
\end{array}
\!\right)_{N}
=
\left(
\begin{array}{cc}
0&0\\
\frac{ic\,\epsilon^{N+1}}{i+c\,\epsilon^N}+ic\,\epsilon^N&0
\end{array}
\right)
\left(\!
\begin{array}{c}
1\\ 0
\end{array}
\!\right)_{N}
\sim\epsilon^N.
\end{align*}
Also there is no divergence in $B_{N-1}(\lambda_1)D_{N-1}(\lambda_2)|0\rangle_{N-1}$.
Thus the first line of (\ref{eq:NepomechieWangAppendix2}) is of order $\epsilon^N$.
In the second line of (\ref{eq:NepomechieWangAppendix2}),
we have
\[
D_{N-1}(\lambda_1)|0\rangle_{N-1}
=\left(
\lambda_1-\frac{i}{2}
\right)^{N-1}|0\rangle_{N-1}
\sim\epsilon^{N-1}
\]
(the first equality follows from the definition of the operator $D_N$, see \cite{Faddeev,KorepinBook}),
and we have
\begin{align*}
\frac{2i+c\,\epsilon^N}{i+c\,\epsilon^N}b_N(\lambda_1)a_N(\lambda_2)
\left(\!
\begin{array}{c}
1\\ 0
\end{array}
\!\right)_{N}
=
\frac{2i+c\,\epsilon^N}{i+c\,\epsilon^N}
\left(
\begin{array}{cc}
0&0\\
i\epsilon&0
\end{array}
\right)
\left(\!
\begin{array}{c}
1\\ 0
\end{array}
\!\right)_{N}
\sim\epsilon.
\end{align*}
Thus the second line of (\ref{eq:NepomechieWangAppendix2}) is also of order $\epsilon^N$.

Summing up, we obtain the relation (\ref{eq:NepomechieWangAppendix_main}).
\end{proof}

Now we need to show that the vector $\lim_{\epsilon\rightarrow 0}\Psi^{(\epsilon)}_\lambda$
is an eigenvector of the Hamiltonian $\mathcal{H}_N$.
Recall the following standard relations (see \cite{Faddeev,KorepinBook})
\begin{align}
\nonumber
&\{A_N(\lambda)+D_N(\lambda)\}B_N(\lambda_1)\cdots B_N(\lambda_\ell)|0\rangle_N\\
&=
\Lambda(\lambda;\lambda_1,\cdots,\lambda_\ell)\prod_{j=1}^\ell B_N(\lambda_j)|0\rangle_N
+\sum_{k=1}^\ell\biggl\{
\Lambda_k(\lambda;\lambda_1,\cdots,\lambda_\ell)B_N(\lambda)
\prod_{j=1\atop j\neq k}^\ell B_N(\lambda_j)|0\rangle_N\biggr\},
\end{align}
where
\begin{align}
\Lambda(\lambda;\lambda_1,\cdots,\lambda_\ell)=
\left(
\lambda+\frac{i}{2}
\right)^N
\prod_{j=1}^\ell \frac{\lambda-\lambda_j-i}{\lambda-\lambda_j}+
\left(
\lambda-\frac{i}{2}
\right)^N
\prod_{j=1}^\ell \frac{\lambda_j-\lambda-i}{\lambda_j-\lambda}
\end{align}
and for $k=1,2,\ldots,k$ we have
\begin{align}
\Lambda_k(\lambda;\lambda_1,\cdots,\lambda_\ell)=
\frac{i}{\lambda-\lambda_k}
\Biggl\{
\left(
\lambda_k+\frac{i}{2}
\right)^N
\prod_{j=1\atop j\neq k}^\ell \frac{\lambda_k-\lambda_j-i}{\lambda_k-\lambda_j}-
\left(
\lambda_k-\frac{i}{2}
\right)^N
\prod_{j=1\atop j\neq k}^\ell \frac{\lambda_j-\lambda_k-i}{\lambda_j-\lambda_k}
\Biggr\}.
\end{align}
As in the standard arguments, we can show
$\Lambda_k(\lambda;\lambda_1,\cdots,\lambda_\ell)=0$
for $k=3,\ldots,\ell$ by the Bethe ansatz equations.
Therefore we need to check the following cases.
\begin{lemma}\label{lem:asymptotics}
Let $\lambda_1=\frac{i}{2}+\epsilon+c\,\epsilon^N$ and $\lambda_2=-\frac{i}{2}+\epsilon$.

(a) If we take
\begin{align}\label{eq:c1_NepomechieWang}
c=-\frac{2}{i^{N+1}}\prod^\ell_{j=3}\frac{\lambda_j-\frac{3i}{2}}{\lambda_j+\frac{i}{2}},
\end{align}
then we have $\Lambda_1(\lambda;\lambda_1,\cdots,\lambda_\ell)=0$.

(b) If we take
\begin{align}\label{eq:c2_NepomechieWang}
c=2i^{N+1}\prod^\ell_{j=3}\frac{\lambda_j+\frac{3i}{2}}{\lambda_j-\frac{i}{2}},
\end{align}
then we have $\Lambda_2(\lambda;\lambda_1,\cdots,\lambda_\ell)=0$.
\end{lemma}
\begin{proof}
(a) We have
\begin{align*}
\Lambda_1&=
\frac{i}{\lambda-\lambda_1}
\Biggl\{
\left(
\lambda_1+\frac{i}{2}
\right)^N
\prod_{j=2}^\ell \frac{\lambda_1-\lambda_j-i}{\lambda_1-\lambda_j}-
\left(
\lambda_1-\frac{i}{2}
\right)^N
\prod_{j=2}^\ell \frac{\lambda_j-\lambda_1-i}{\lambda_j-\lambda_1}
\Biggr\}\\
&=
\frac{i}{\lambda-\lambda_1}
\Biggl\{
i^N\cdot\frac{c\,\epsilon^N}{i}
\prod_{j=3}^\ell \frac{\frac{i}{2}-\lambda_j-i}{\frac{i}{2}-\lambda_j}-
\epsilon^N\cdot\frac{-2i}{-i}
\prod_{j=3}^\ell \frac{\lambda_j-\frac{i}{2}-i}{\lambda_j-\frac{i}{2}}
\Biggr\}\\
&=
\frac{i\,\epsilon^N}{\lambda-\lambda_1}
\Biggl\{
c\cdot i^{N-1}
\prod_{j=3}^\ell \frac{\lambda_j+\frac{i}{2}}{\lambda_j-\frac{i}{2}}-
2
\prod_{j=3}^\ell \frac{\lambda_j-\frac{3i}{2}}{\lambda_j-\frac{i}{2}}
\Biggr\}.
\end{align*}
Therefore if we take $c$ as in (\ref{eq:c1_NepomechieWang}),
we see that $\Lambda_1=0$.

(b) We have
\begin{align*}
\Lambda_2&=
\frac{i}{\lambda-\lambda_2}
\Biggl\{
\left(
\lambda_2+\frac{i}{2}
\right)^N
\prod_{j=1\atop j\neq 2}^\ell \frac{\lambda_2-\lambda_j-i}{\lambda_2-\lambda_j}-
\left(
\lambda_2-\frac{i}{2}
\right)^N
\prod_{j=1\atop j\neq 2}^\ell \frac{\lambda_j-\lambda_2-i}{\lambda_j-\lambda_2}
\Biggr\}\\
&=
\frac{i}{\lambda-\lambda_2}
\Biggl\{
\epsilon^N\cdot\frac{-2i}{-i}
\prod_{j=3}^\ell \frac{-\frac{i}{2}-\lambda_j-i}{-\frac{i}{2}-\lambda_j}-
(-i)^N\,\frac{c\,\epsilon^N}{i}
\prod_{j=3}^\ell \frac{\lambda_j+\frac{i}{2}-i}{\lambda_j+\frac{i}{2}}
\Biggr\}\\
&=
\frac{i\,\epsilon^N}{\lambda-\lambda_2}
\Biggl\{
2\prod_{j=3}^\ell \frac{\lambda_j+\frac{3i}{2}}{\lambda_j+\frac{i}{2}}-
\frac{c}{i^{N+1}}
\prod_{j=3}^\ell \frac{\lambda_j-\frac{i}{2}}{\lambda_j+\frac{i}{2}}
\Biggr\}.
\end{align*}
Therefore if we take $c$ as in (\ref{eq:c2_NepomechieWang}),
we see that $\Lambda_2=0$.
\end{proof}

Now we can prove the theorem of Nepomechie--Wang
\cite{NepomechieWang2013}.

\begin{theorem}
Suppose that the singular solution $\lambda=\{\lambda_1=\frac{i}{2},
\lambda_2=-\frac{i}{2},\lambda_3,\ldots,\lambda_\ell\}$ satisfies the condition
\begin{align}
\label{eq:NepomechieWangCondition}
\left(
-\prod^\ell_{j=3}\frac{\lambda_j+\frac{i}{2}}{\lambda_j-\frac{i}{2}}
\right)^N=1.
\end{align}
Then the vector $\lim_{\epsilon\rightarrow 0}\Psi^{(\epsilon)}_\lambda$ is a
well-defined eigenvector of the Hamiltonian $\mathcal{H}_N$.
\end{theorem}
\begin{proof}
The remaining thing to check is that the compatibility condition of
(\ref{eq:c1_NepomechieWang}) and (\ref{eq:c2_NepomechieWang})
is the condition (\ref{eq:NepomechieWangCondition}).
By equating the right hand sides of (\ref{eq:c1_NepomechieWang})
and (\ref{eq:c2_NepomechieWang}), we obtain
\[
\prod_{j=3}^\ell\frac{\lambda_j-\frac{i}{2}}{\lambda_j+\frac{i}{2}}\cdot
\frac{\lambda_j-\frac{3i}{2}}{\lambda_j+\frac{3i}{2}}
=(-1)^N.
\]
On the other hand, if we take product of the Bethe ansatz equations,
we obtain
\[
\prod_{j=3}^\ell
\left(\frac{\lambda_k+\frac{i}{2}}{\lambda_k-\frac{i}{2}}\right)^{N-1}
\left(\frac{\lambda_k-\frac{3i}{2}}{\lambda_k+\frac{3i}{2}}\right)
=\prod_{k,j=3 \atop k\neq j}^{\ell}
\frac{\lambda_k-\lambda_j+i}{\lambda_k-\lambda_j-i}=1.
\]
Combining two expressions, we obtain (\ref{eq:NepomechieWangCondition}).
\end{proof}

Finally we have to compute the energy eigenvalue for
the state $\lim_{\epsilon\rightarrow 0}\Psi^{(\epsilon)}_\lambda$.

\begin{proposition}[\cite{KirillovSakamoto2014b}]
For the singular solution $\lambda=\{\lambda_1=\frac{i}{2},
\lambda_2=-\frac{i}{2},\lambda_3,\ldots,\lambda_\ell\}$,
the energy eigenvalue of the corresponding vector
$\lim_{\epsilon\rightarrow 0}\Psi^{(\epsilon)}_\lambda$ is
\[
\mathcal{E}=-J-\frac{J}{2}\sum_{j=3}^\ell\frac{1}{\lambda_j^2+\frac{1}{4}}.
\]
\end{proposition}
\begin{proof}
Let $\varepsilon=\frac{2}{J}\mathcal{E}+N$.
Then
\begin{align}
\varepsilon=i\frac{d}{d\lambda}\log\Lambda\biggr|_{\lambda=\frac{i}{2}}
=\frac{i\frac{d\Lambda}{d\lambda}}{\Lambda}\biggr|_{\lambda=\frac{i}{2}},
\end{align}
where
\begin{align}
\Lambda(\lambda;\lambda_1,\cdots,\lambda_\ell)=
\left(
\lambda+\frac{i}{2}
\right)^N
\prod_{j=1}^\ell \frac{\lambda-\lambda_j-i}{\lambda-\lambda_j}+
\left(
\lambda-\frac{i}{2}
\right)^N
\prod_{j=1}^\ell \frac{\lambda_j-\lambda-i}{\lambda_j-\lambda}.
\end{align}
See (\ref{eq:tau_hamiltonian}).
On this expression, we apply the regularization
\[
\lambda_1=\frac{i}{2}+\epsilon+c\,\epsilon^N,\qquad
\lambda_2=-\frac{i}{2}+\epsilon.
\]

The denominator of $\varepsilon$ is
\begin{align*}
\varepsilon_{\rm deno}:=&\,
\Lambda\left(\frac{i}{2};\lambda_1,\cdots,\lambda_\ell\right)=
i^N\prod_{j=1}^\ell\frac{\lambda_j+\frac{i}{2}}{\lambda_j-\frac{i}{2}}\\
=&\,i^N\cdot
\frac{i+\epsilon+c\,\epsilon^N}{\epsilon+c\,\epsilon^N}\cdot
\frac{\epsilon}{\epsilon-i}\,
\prod_{j=3}^\ell\frac{\lambda_j+\frac{i}{2}}{\lambda_j-\frac{i}{2}}
=i^N\cdot
\frac{i+\epsilon+c\,\epsilon^N}{(1+c\,\epsilon^{N-1})(\epsilon-i)}\,
\prod_{j=3}^\ell\frac{\lambda_j+\frac{i}{2}}{\lambda_j-\frac{i}{2}}.
\end{align*}

On the other hand, the numerator is
\begin{align}
\nonumber
i\frac{d\Lambda}{d\lambda}=&\,
A_0(\lambda)+
\sum^\ell_{j=1}A_j(\lambda)+\text{terms containing at least one}\left(\lambda-\frac{i}{2}\right)
\end{align}
where
\begin{align*}
A_0(\lambda)=iN\!\left(\lambda+\frac{i}{2}\right)^{N-1}
\prod_{j=1}^\ell\frac{\lambda-\lambda_j-i}{\lambda-\lambda_j}
\end{align*}
and for $j=1,2,\ldots,\ell$,
\begin{align*}
A_j(\lambda)=
i\!\left(\lambda+\frac{i}{2}\right)^{N}
\frac{\lambda-\lambda_1-i}{\lambda-\lambda_1}\cdots
\frac{\lambda-\lambda_{j-1}-i}{\lambda-\lambda_{j-1}}\cdot
\frac{i}{(\lambda_j-\lambda)^2}\cdot
\frac{\lambda-\lambda_{j+1}-i}{\lambda-\lambda_{j+1}}\cdots
\frac{\lambda-\lambda_\ell-i}{\lambda-\lambda_\ell}.
\end{align*}
Below we compute the contribution from each term one by one.

$\bullet$ Let us consider $A_0(\lambda)$:
\begin{align*}
A_0\!\left(\frac{i}{2}\right)&=
i^NN\cdot
\frac{\frac{i}{2}-(\frac{i}{2}+\epsilon+c\,\epsilon^N)-i}{\frac{i}{2}-(\frac{i}{2}+\epsilon+c\,\epsilon^N)}\cdot
\frac{\frac{i}{2}-(-\frac{i}{2}+\epsilon)-i}{\frac{i}{2}-(-\frac{i}{2}+\epsilon)}\,
\prod_{j=3}^\ell\frac{\lambda_j+\frac{i}{2}}{\lambda_j-\frac{i}{2}}\\
&=i^NN\cdot
\frac{i+\epsilon+c\,\epsilon^N}{(1+c\,\epsilon^{N-1})(\epsilon-i)}\,
\prod_{j=3}^\ell\frac{\lambda_j+\frac{i}{2}}{\lambda_j-\frac{i}{2}}.
\end{align*}
Therefore we obtain
\begin{align*}
\frac{1}{\varepsilon_{\rm deno}}\cdot
A_0\!\left(\frac{i}{2}\right)=N.
\end{align*}

$\bullet$ Let us consider $A_1(\lambda)$ and $A_2(\lambda)$.
\begin{align*}
A_1\!\left(\frac{i}{2}\right)&=
i^{N+1}
\frac{i}{\{\frac{i}{2}-(\frac{i}{2}+\epsilon+c\,\epsilon^N)\}^2}\cdot
\frac{\frac{i}{2}-(-\frac{i}{2}+\epsilon)-i}{\frac{i}{2}-(-\frac{i}{2}+\epsilon)}\,
\prod_{j=3}^\ell\frac{\lambda_j+\frac{i}{2}}{\lambda_j-\frac{i}{2}}\\
&=-i^{N}
\frac{1}{\epsilon\,(1+c\,\epsilon^{N-1})^2(\epsilon-i)}\,
\prod_{j=3}^\ell\frac{\lambda_j+\frac{i}{2}}{\lambda_j-\frac{i}{2}}.
\end{align*}
On the other hand, we have
\begin{align*}
A_2\!\left(\frac{i}{2}\right)&=
i^{N+1}\,
\frac{\frac{i}{2}-(\frac{i}{2}+\epsilon+c\,\epsilon^N)-i}{\frac{i}{2}-(\frac{i}{2}+\epsilon+c\,\epsilon^N)}\cdot
\frac{i}{\{\frac{i}{2}-(-\frac{i}{2}+\epsilon)\}^2}\,
\prod_{j=3}^\ell\frac{\lambda_j+\frac{i}{2}}{\lambda_j-\frac{i}{2}}\\
&=-i^{N}
\frac{i+\epsilon+c\,\epsilon^N}{\epsilon\,(1+c\,\epsilon^{N-1})(\epsilon-i)^2}\,
\prod_{j=3}^\ell\frac{\lambda_j+\frac{i}{2}}{\lambda_j-\frac{i}{2}}.
\end{align*}
Thus we have
\begin{align*}
&\lim_{\epsilon\rightarrow 0}\,
\frac{1}{\varepsilon_{\rm deno}}
\left\{
A_1\!\left(\frac{i}{2}\right)+A_2\!\left(\frac{i}{2}\right)
\right\}\\
=&\lim_{\epsilon\rightarrow 0}\,
\frac{1}{\varepsilon_{\rm deno}}\times
\left(-i^{N}\right)
\frac{(\epsilon-i)+(i+\epsilon+c\,\epsilon^N)(1+c\,\epsilon^{N-1})}
{\epsilon\,(1+c\,\epsilon^{N-1})^2(\epsilon-i)^2}\,
\prod_{j=3}^\ell\frac{\lambda_j+\frac{i}{2}}{\lambda_j-\frac{i}{2}}\\
=&-\lim_{\epsilon\rightarrow 0}\,
\frac{(1+c\,\epsilon^{N-1})(\epsilon-i)}{i+\epsilon+c\,\epsilon^N}\times
\frac{2\epsilon+i\,c\,\epsilon^{N-1}+2c\,\epsilon^N+c^2\epsilon^{2N-1}}
{\epsilon\,(1+c\,\epsilon^{N-1})^2(\epsilon-i)^2}
\\
=&-\lim_{\epsilon\rightarrow 0}\,
\frac{2\epsilon+i\,c\,\epsilon^{N-1}+2c\,\epsilon^N+c^2\epsilon^{2N-1}}
{\epsilon\,(1+c\,\epsilon^{N-1})(\epsilon-i)(i+\epsilon+c\,\epsilon^N)}
=-2.
\end{align*}

$\bullet$ Finally, for $j=3,4,\ldots,\ell$, we have
\begin{align*}
A_j\!\left(\frac{i}{2}\right)=&\,
i^{N+1}\,\frac{i+\epsilon+c\,\epsilon^N}{\epsilon+c\,\epsilon^N}\cdot
\frac{\epsilon}{\epsilon-i}\cdot
\frac{\lambda_3+\frac{i}{2}}{\lambda_3-\frac{i}{2}}\cdots
\frac{\lambda_{j-1}+\frac{i}{2}}{\lambda_{j-1}-\frac{i}{2}}\cdot
\frac{i}{(\lambda_j-\frac{i}{2})^2}\cdot
\frac{\lambda_{j+1}+\frac{i}{2}}{\lambda_{j+1}-\frac{i}{2}}\cdots
\frac{\lambda_\ell+\frac{i}{2}}{\lambda_\ell-\frac{i}{2}}.
\end{align*}
Thus we have
\begin{align*}
\frac{1}{\varepsilon_{\rm deno}}\cdot
A_j\!\left(\frac{i}{2}\right)=
-\,\frac{\lambda_j-\frac{i}{2}}{\lambda_j+\frac{i}{2}}\cdot
\frac{1}{(\lambda_j-\frac{i}{2})^2}=
-\,\frac{1}{\lambda_j^2+\frac{1}{4}}.
\end{align*}

To summarize, we have
\begin{align*}
\varepsilon=N-2-\sum_{j=3}^\ell\frac{1}{\lambda_j^2+\frac{1}{4}}.
\end{align*}
Therefore we obtain the result.
\end{proof}

\section{Update for our previous paper}
\label{sec:update}
Here we would like to provide an update for our previous paper \cite{KiSa14}.
In $N=12$, $\ell=5$ case, we have the following solution:
\begin{align*}
\lambda_1&=0\\
\lambda_2&=0.178978221719006005297692789210\cdots\\
\lambda_3&=-0.178978221719006005297692789210\cdots\\
\lambda_4&=i/2\\
\lambda_5&=-i/2
\end{align*}
Here, $\xi=\lambda_2^2=\lambda_3^2$
is the smallest positive real solution of the equation
\begin{align}
\label{eq:N12M5}
5120 \xi ^5+11520 \xi ^4-4992 \xi ^3-9312 \xi ^2+2020 \xi -55=0.
\end{align}
This solution corresponds to the following rigged configuration:
\begin{center}
\unitlength 12pt
\begin{picture}(4,4)
\put(0,3.1){2}
\put(0,2.1){4}
\put(0,1.1){4}
\put(0,0.1){4}
\put(0.8,0){\Yboxdim12pt\yng(2,1,1,1)}
\put(3.0,3.1){1}
\put(2.1,2.1){2}
\put(2.1,1.1){2}
\put(2.1,0.1){2}
\end{picture}
\end{center}
In ``tableN12M5.pdf" of \cite{HNS1}, this solution appears as \#235
which is counted as a regular solution.
Based on this result, we noted in \cite[Conjecture 14(C)]{KiSa14}
such that we need one more restriction on the riggings when $\ell$ is odd.
However, as Deguchi--Giri \cite{DG} pointed out, the above solution is singular.
Therefore we come to the conclusion that we no longer need restriction \cite[Eq.(26)]{KiSa14}
when $\ell$ is odd.

We remark that equation (\ref{eq:N12M5}) admits two more positive real solutions and
both of them provide physical singular solutions.
In the rigged configurations language they are $\{(2,1),(1,3),$ $(1,2),(1,1)\}$ for the smaller solution
and $\{(2,1),(1,4),(1,2),(1,0)\}$ for the larger solution.
Similarly the equation (\ref{eq:N12M5}) admits two negative real solutions
which provide physical singular solutions.
In the rigged configurations language, they are $\{(4,1),(1,4)\}$ for the smaller solution
and $\{(3,1),(2,2)\}$ for the larger solution
Thus the five roots of equation (\ref{eq:N12M5}) determine all the physical singular solutions
for the case $N=12$ and $\ell=5$.

\paragraph{Acknowledgments:}
I would like to thank Prof. Anatol N. Kirillov for the collaboration in the related works
\cite{KiSa14,KirillovSakamoto2014b} and kind interest in the present paper.


\begin{thebibliography}{99}
\bibitem[AV]{AV}
L.~V.~Avdeev and A.~A.~Vladimirov,
{\it Exceptional solutions to the Bethe ansatz equations},
Theor. Math. Phys. {\bf 69} (1986) 1071--1079.

\bibitem[BMSZ]{Beisert}
N.~Beisert, J.~A.~Minahan, M.~Staudacher and K.~Zarembo,
{\it Stringing spins and spinning strings},
JHEP {\bf 09} (2003) 010 (27pp).

\bibitem[B]{Bethe}
H.~A.~Bethe,
{\it Zur theorie der metalle},
Zeit. f\"{u}r Physik {\bf 71} (1931) 205--226.

\bibitem[D]{Deguchi2001} T.~Deguchi,
{\it Non-regular eigenstate of the XXX model as some limit of the Bethe state},
J. Phys. A: Math. Gen. {\bf 34} (2001) 9755--9775.

\bibitem[DG]{DG} T.~Deguchi and P.~R.~Giri, 
\textit{Non self-conjugate strings, singular strings and rigged configurations in the Heisenberg model},
 J. Stat. Mech: Theor. Exp. (2015) P02004.

\bibitem[EKS]{EKS:1992}
F.~H.~L.~Essler, V.~E.~Korepin and K.~Schoutens,
{\it Fine structure of the Bethe ansatz for the
spin-$\frac{1}{2}$ Heisenberg $XXX$ model},
J. Phys. A: Math. Gen. {\bf 25} (1992) 4115--4126.

\bibitem[F]{Faddeev}
L.~D.~Faddeev,
{\it How Algebraic Bethe Ansatz works for integrable model},
arXiv:hep-th/9605187 (Les-Houches lectures).

\bibitem[GD1]{GD}
P.~R.~Giri and T.~Deguchi,
{\it Singular eigenstates in the even(odd) length Heisenberg spin chain},
arXiv:1411.5839

\bibitem[GD2]{GD2}
P.~R.~Giri and T.~Deguchi,
{\it Heisenberg model and rigged configurations},
arXiv:1501.07801

\bibitem[HC]{HC}
R.~Hagemans and J.-S.~Caux,
{\it Deformed strings in the Heisenberg model},
J. Phys. A: Math. Theor. {\bf 40} (2007) 14605--14647.

\bibitem[HNS]{HNS1}
W.~Hao, R.~I.~Nepomechie and A.~I.~Sommese,
{\it Completeness of solutions of Bethe's equations},
Phys. Rev. E {\bf 88} (2013) 052113 (8pp plus supplemental material).

\bibitem[KKR]{KKR}
S.~V.~Kerov, A.~N.~Kirillov, and N.~Yu~Reshetikhin,
{\it Combinatorics, Bethe ansatz, and representations of the symmetric group}.
Zap. Nauch. Sem. LOMI {\bf 155} (1986) 50--64.
(English Translation: J. Soviet Math. {\bf 41} (1988) 916--924.)

\bibitem[KR]{KR}
A.~N.~Kirillov and N.~Yu.~Reshetikhin:
{\it The Bethe ansatz and the combinatorics of Young tableaux},
Zap. Nauch. Sem. LOMI {\bf 155} (1986) 65--115.
(English Translation: J. Soviet Math. {\bf 41} (1988) 925--955.)

\bibitem[KS14a]{KiSa14}
A.~N.~Kirillov and R.~Sakamoto,
{\it Singular solutions to the Bethe ansatz equations and rigged configurations},
J. Phys. A: Math. Theor. {\bf 47} (2014) 205207 (20pp).

\bibitem[KS14b]{KirillovSakamoto2014b}
A.~N.~Kirillov and R.~Sakamoto,
{\it Some remarks on Nepomechie--Wang eigenstates for spin 1/2 XXX model},
Moscow Math. J. (to appear, arXiv:1406.1958).

\bibitem[KBI]{KorepinBook}
V.~E.~Korepin, N.~M.~Bogoliubov and A.~G.~Izergin,
{\it Quantum inverse scattering method and correlation functions},
Cambridge University Press (1993).

\bibitem[NW13]{NepomechieWang2013}
R.~I.~Nepomechie and C.~Wang,
{\it Algebraic Bethe ansatz for singular solutions},
J. Phys. A: Math. Theor. {\bf 46} (2013) 325002 (8pp).

\bibitem[NW14]{NepomechieWang2014}
R.~I.~Nepomechie and C.~Wang,
{\it Twisting singular solutions of Bethe's equations},
J. Phys. A: Math. Theor. {\bf 47} (2014) 505004 (9pp).

\bibitem[SD]{SD}
J.~Sato and T.~Deguchi,
{\it Numerical analysis of string solutions of the integrable XXZ spin chains},
presentation available at

\texttt{http://cfim11.sciencesconf.org/conference/cfim11/Dijon\_JS.pdf}

\bibitem[T]{Takahashi}
M.~Takahashi,
{\it One-dimensional Heisenberg model at finite temperature},
Prog. Theor. Phys. {\bf 46} (1971) 401--415.

\end{thebibliography}
\end{document}